\documentclass[a4paper,11pt]{amsart}
\usepackage{tikz}
\usetikzlibrary{matrix}
\usepackage[small,nohug,heads=vee]{diagrams}
\diagramstyle[labelstyle=\scriptstyle]
\usepackage{pst-node}
\usepackage{amsmath,yhmath,mathrsfs}
\usepackage{color}
\usepackage{comment}
\usepackage{amssymb}
\usepackage{amsfonts}
\usepackage{youngtab}
\usepackage{amscd}
\usepackage{slashed}
\usepackage{pdflscape}
\usepackage{latexsym}
\usepackage{multirow}
\usepackage{stmaryrd}
\usepackage[backref,pagebackref,linkcolor=blue]{hyperref}
\usepackage{mathrsfs}

\newcommand{\bm}[1]{\mbox{\boldmath~$ #1~$}}
\newcommand{\ce}{{\Cal E}}
\newtheorem{theorem}{Theorem}[section]
\newtheorem{lemma}[theorem]{Lemma}
  
\newtheorem{proposition}[theorem]{Proposition}

\theoremstyle{definition}
\newtheorem{definition}[theorem]{Definition}

\theoremstyle{remark}
\newtheorem{remark}[theorem]{Remark}

\usepackage{graphicx} 
\usepackage{amsmath} 
\usepackage{amsfonts}
\usepackage{amssymb}

\newcommand{\be}{\begin{equation}}

\newcommand{\ee}{\end{equation}}

\renewcommand{\c}{\boldsymbol{c}}
\newcommand{\p}{\boldsymbol{p}}
\newcommand{\bw}{\Wedge}

\newcommand{\W}{\mathcal{W}}

\newcommand{\om}{\omega}

\newcommand{\cB}{{\mathcal B}}

\newcommand{\cG}{{\mathcal G}}
\newcommand{\cU}{\mathcal{U}}

\newcommand{\si}{\sigma}

\newcommand{\ba}{\begin{array}}

\newcommand{\ea}{\end{array}}

\newcommand{\beq}{\begin{eqnarray}}

\newcommand{\eeq}{\end{eqnarray}}

\newtheorem{lm}{lemma}

\newtheorem{thee}{theorem}

\newtheorem{proo}{proposition}

\newtheorem{co}{corollary}

\newtheorem{rem}{remark}

\newtheorem{deff}{definition}

\newcommand{\bd}{\begin{deff}}

\newcommand{\ed}{\end{deff}}

\newcommand{\bl}{\begin{lm}}

\newcommand{\el}{\end{lm}}

\newcommand{\bp}{\begin{proo}}

\newcommand{\ep}{\end{proo}}

\newcommand{\bt}{\begin{thee}}

\newcommand{\et}{\end{thee}}

\newcommand{\bc}{\begin{co}}

\newcommand{\ec}{\end{co}}

\newcommand{\brm}{\begin{rem}}

\newcommand{\erm}{\end{rem}}

\usepackage{t1enc}
\def\frak{\mathfrak}

\def\Cal{\mathcal}
\newcommand{\g}{\mathfrak{g}}
\newcommand{\cL}{{\Cal L}}

\newcommand{\newc}{\newcommand}

\let\ccdot\cdot
\def\cdot{\hbox to 2.5pt{\hss$\ccdot$\hss}}

\newc{\aR}{\mbox{\boldmath{$ R$}}}
\newc{\aS}{\mbox{\boldmath{$ S$}}}
\newc{\aT}{\mbox{\boldmath{$ T$}}}
\newc{\aW}{\mbox{\boldmath{$ W$}}}

\newc{\aD}{\mbox{\boldmath{$ D$}}\hspace{-.2mm}}

\newc{\aK}{\mbox{\boldmath{$ K$}}}
\newc{\aL}{\mbox{\boldmath{$ L$}}}

\newcommand{\ct}{{\Cal T}}

\usepackage{amssymb}
\usepackage{amscd}

\newcommand{\nd}{\nabla}

\newcommand{\cT}{{\mathcal T}}

\newcommand{\cF}{{\Cal F}}
\newcommand{\cV}{{\Cal V}}

\newcommand{\nn}[1]{(\ref{#1})}

\def\cD{{\Cal D}}

\newc{\obstrn}[2]{B^{#1}_{#2}}

\newcommand{\flplus} 
{
\hspace{0.1cm}
\begin{tikzpicture}[baseline=-0.582ex]
    \draw [line width=0.24pt](-0.1129, 0) -- (0.1129, 0) -- (0, 0) -- (0, 0.1129) -- (0, -0.1129) arc (270:90:0.1129) -- (0, 0);
\end{tikzpicture}
\hspace{0.1cm}
}

\newcommand{\frplus} 
{
\hspace{0.1cm}
\begin{tikzpicture}[baseline=-0.582ex]
    \draw [line width=0.24pt](-0.1129, 0) -- (0.1129, 0) -- (0, 0) -- (0, 0.1129) -- (0, -0.1129) arc (-90:90:0.1129) -- (0, 0);
\end{tikzpicture}
\hspace{0.1cm}
}

\usepackage{ifthen}

\newc{\tensor}[1]{#1}
\newc{\Mvariable}[1]{\mbox{#1}}
\newc{\down}[1]{{}_{#1}}
\newc{\up}[1]{{}^{#1}}

\newc{\JulyStrut}{\rule{0mm}{6mm}}
\newc{\midtenPan}{\mbox{\sf S}}
\newc{\midten}{\mbox{\sf T}}
\newc{\midtenEi}{\mbox{\sf U}}
\newc{\ATen}{\mbox{\sf E}}
\newc{\BTen}{\mbox{\sf F}}
\newc{\CTen}{\mbox{\sf G}}

\def\sideremark#1{\ifvmode\leavevmode\fi\vadjust{\vbox to0pt{\vss
 \hbox to 0pt{\hskip\hsize\hskip1em
 \vbox{\hsize3cm\tiny\raggedright\pretolerance10000
 \noindent #1\hfill}\hss}\vbox to8pt{\vfil}\vss}}}%
                        
\numberwithin{equation}{section}

\newcommand{\B}{\mathcal B}

\newcommand{\faketop}{{\phantom{\scalebox{.6}{\textrm{top}}}\!\!\!}}

\renewcommand{\P}{{\mbox{\sf P}}}                   

\renewcommand{\W}{{\mbox{\sf W}}}                   

\newcommand{\bes}{\begin{equation*}}
\newcommand{\ees}{\end{equation*}}

\newcommand{\Wedge}{\scalebox{1.25}{$\wedge$}\!}

\newcommand{\partialf}{\boldsymbol{\partial}^*}
\newcommand{\deltaf}{\boldsymbol{\delta}}
\newcommand{\deltafs}{\scalebox{.7}{$\boldsymbol {\delta}$}}
\newcommand{\df}{\boldsymbol{d}}
\newcommand{\dfs}{\scalebox{.7}{$\boldsymbol{d}$}}
\newcommand{\Ms}{\scalebox{.7}{$\boldsymbol{ M}$}}
\newcommand{\Vf}{{\!\bm V\!\!}}

\begin{document}

\title{Metric projective geometry, BGG detour complexes\\[.5mm] and partially massless gauge theories}
\author{A. Rod Gover${}^\diamondsuit$, Emanuele Latini${}^\clubsuit$ \& Andrew Waldron${}^\spadesuit$}

\address{${}^\diamondsuit$Department of Mathematics\\
  The University of Auckland\\
  Private Bag 92019\\
  Auckland 1\\
  New Zealand,  and 
  Mathematical Sciences Institute, Australian National University, ACT
  0200, Australia} \email{gover@math.auckland.ac.nz}
  
  \address{${}^{\clubsuit}$Institut f{\"u}r Mathematik, Universit{\"a}t Z{\"u}rich-Irchel, Winterthurerstrasse 190, CH-8057 Z{\"u}rich, 
  Switzerland, and INFN, Laboratori Nazionali di Frascati, CP 13,
  I-00044 Frascati, Italy} \email{emanuele.latini@math.uzh.ch}
  
  \address{${}^{\spadesuit}$Department of Mathematics\\
  University of California\\
  Davis, CA95616, USA} \email{wally@math.ucdavis.edu}

\vspace{10pt}

\vspace{10pt}

\renewcommand{\arraystretch}{1}


\begin{abstract} 

A projective geometry is an equivalence class of torsion free
connections sharing the same unparametrised geodesics; this is a basic
structure for understanding physical systems.  Metric projective
geometry is concerned with the interaction of projective and
pseudo-Riemannian geometry. 
We show that the BGG machinery
of projective geometry combines with structures known as Yang-Mills
detour complexes to produce a general tool for generating invariant
pseudo-Riemannian gauge theories.  This produces (detour) complexes of
differential operators corresponding to gauge invariances and
dynamics.  We show, as an
application, that curved versions of these sequences give geometric characterizations of the obstructions to
propagation of higher spins in Einstein spaces.  Further, we show that
projective BGG detour complexes generate both gauge invariances and
gauge invariant constraint systems for partially massless models: the
input for this machinery is a projectively invariant gauge operator
corresponding to the first operator of a certain BGG sequence. We also
connect this technology to the log-radial reduction method and extend
the latter to Einstein backgrounds.

\vspace{2cm}

\end{abstract}

\maketitle

\pagestyle{myheadings} \markboth{Gover, Latini \& Waldron}{Metric projective geometry}

\newpage

\tableofcontents

\section{Introduction}\label{intro}

A fundamental problem in physics, mathematics, and their interface is
that of finding the right way to describe and treat the natural
differential equations describing fields.  Naturality here refers to
equations determined by the underlying geometry, so this is in essence
a geometric problem. For the case of particle theories in space-time
physics it appears, on the surface, that
pseudo-Riemannian geometry should be the central structure. However
pseudo-Riemannian invariance is a rather weak condition, in the sense
that by far more equations are invariant in this sense than are
interesting or important.  For field equations of motion, other
principles can be brought to bear, and in particular the requirement
that theories exhibit suitable gauge invariance, and corresponding
integrability conditions, plays a critical {\it r\^ole} in determining systems
with, for example, the correct propagating degrees of freedom (DoF). It is
then reasonable to ask if there is a more fundamental geometric structure, underlying
 pseudo-Riemmannian geometry, that includes the entire
picture. The payoff for a positive answer can be significant. Apart
from developing theory for the unification and extension of gauge
theories, this can give insight into how to treat
fields at infinity and  thus understand decay, scattering, and possible
holographic features.

It has long been realised that, in dimension four, the equations
governing massless fields exhibit conformal invariance. This suggests
conformal geometry (a manifold $M$ and a conformal class of
metrics~${\c}$) as a central organizing principle for field
theories~\cite{Gover:2008sw,Gover:2008pt}.  However four dimensional
physics also requires fields that are not massless, and so one may consider
alternative basic structures, such as a projective class of
connections: Given an affine connection~$\nabla$, the collection of
its geodesics as unparametrised curves is the corresponding {\em
  projective structure} $\p$; or $\p$ may be viewed as the equivalence
class of torsion free connections sharing the same unparametrised
geodesics. To each pseudo-Riemannian metric, there is associated a unique
projective structure via the metric's Levi-connection (or its
geodesics). Then, {\em metric projective geometry} studies the
interaction between these geometries.  Building partly on earlier
works, {\it e.g.}~\cite{Mikes,Sinjukov}, there has been a recent surge
of interest in the links between metric and projective geometry with
powerful results obtained, see {\it e.g.}~\cite{KM,Matveev}.  Through
its projective structure a metric determines a projective Cartan
connection or, equivalently, tractor connection~\cite{BEG,C,T}. This
is a higher order geometric structure which encodes geometry, at each
point, not just in the tangent space but also in higher order Taylor
series data of the manifold. Furthermore it exposes deeper links
between metrics and projective geometry~\cite{EM}. In particular,
there  is a striking connection between projective geometry and Einstein
metrics \cite{A,CGM,GM}. The latter are generalised through suitable
Cartan holonomy reductions~\cite{ageom,CGHduke}. These new insights
have led to the development of a metric-projective analogue of
conformal compactification~\cite{Cap:2013dva}.

Geodesics are basic geometric structures; physically they encode 
how particles interact with background geometries.
Projective geometry is also intimately related to massless spin two dynamics: 
On constant curvature backgrounds, the differential complex controlling
deformations of Riemannian structure (metric fluctuations) is a projectively invariant Bernstein-Gelfand-Gelfand (BGG) complex--see Eastwood's interpretation~\cite{elasticity} of the Calabi complex~\cite{Calabi}. 

BGG complexes have their origins in representation
theory~\cite{BGG,Zuckermann}, but in geometry the dual structures that
go by the same name (or {\em BGG sequences} more generally) arise
naturally from a tractor connection twisting of the de Rham complex
and algebraic tools of
Kostant~\cite{Kostant,Baston,Calderbank-Diemer,CSS-BGG}.  These are
extremely powerful tools for the organisation and interpretation of
invariant differential operators but, as we shall explain in Section
~\ref{gaugeS}, are not the right objects for dynamical field theories.
In the setting of even dimensional conformal geometry, {\em detour
  complexes} were introduced in~\cite{BG1,BG2} as complexes which are
linked to the BGG sequences but which, importantly, involve weaker
integrability conditions. It was quickly realised that these fit into a
wider framework which  fits closely with variational principles.  In
particular, in \cite{GSS} it is shown that, for each linear vector
bundle connection, there is a corresponding differential (detour)
sequence that forms a complex if and only if the given connection satisfies
the Yang-Mills equations; for the vector bundle corresponding to the
adjoint representation of the gauge group the sequence governs second
variations of the Yang-Mills action. These {\em Yang-Mills detour
  complexes} can be linked to the first operators in BGG sequences
(and their adjoints) via differential splitting operators to yield
further complexes that may be called {\it BGG detour complexes}.

In this article, we show that bringing together the BGG machinery and
the Yang-Mills detour theory, in the setting of metric projective
geometry, produces a general tool for generating invariant
pseudo-Riemannian gauge theories. This yields the underlying 
geometric picture we were seeking.  In particular in Theorem~\ref{Bigtheorem}, for fields of any integral spin on constant
curvature backgrounds, we construct equations of motion that are
invariant  with respect to maximal depth ({\it i.e.},\ highest possible derivative order) gauge transformations. The
detour complexes involved also give the corresponding Bianchi
identities. In fact, more than the complex arises: in Theorem~\ref{contraintSys} we show that the detour machinery also produces the
gauge invariant constraint systems and corresponding relations between
the constraints. Thus one obtains a rather complete picture for what
are called, following~\cite{Deser:2001us}, (maximal depth) {\em partially
  massless fields}.

To illustrate  concretely by example  the link between partially massless (PM) theories and projective BGG sequences, consider 
the formula
$$
\delta \varphi_{ab}=\big(\nabla_a\nabla_b+\frac{\Lambda}{3}\, g_{ab}\big) \alpha\, .
$$ 
Physically the above is the gauge invariance of a PM spin two field.  From a geometry perspective, the terms in
brackets are a projectively invariant BGG operator, but specialised to
an Einstein scale. There have also been indications of the importance
of projective geometry for PM systems in the physics literature: It
was first observed in~\cite{Hallowell} that  PM models have as a
geometric origin massless models in a flat space of one higher
dimension. This was achieved by a log-radial reduction~\cite{Biswas}
which projectivised this flat ambient space.  This also suggests an
intimate link between projective geometry and PM models.  Further
evidence for this, especially in light of the conformal-projective link on Einstein manifolds~\cite{GM}, is that these models can also be described by conformally invariant equations coupled to a parallel, conformal, scale tractor~\cite{Gover:2008sw,Gover:2008pt,Grigoriev:2011gp}. Since PM models are our running example of a physical system whose underpinning is a projective structure, let us briefly review those models:
The first PM  theory was discovered by Deser and Nepomechie who were searching for a modification of the spin two equations of motion that supported lightcone propagation in conformally flat spacetimes~\cite{Deser:1983tm}. Subsequently Higuchi realized that this gives a unitarity bound on massive spin two excitations~\cite{Higuchi}.
Later still, it was understood that PM  theories existed for all spins and also Fermi fields~\cite{Deser:2001pe}. Their novel, higher derivative, gauge invariances implied they described the lightlike propagation~\cite{Deser:2001xr} of sets of helicity states that were intermediate between the usual massless and massive models~\cite{Deser:2001pe,Deser:2001us,Deser:2001wx,Deser:2003gw}.

In more detail, massless higher spin systems are described by second
derivative order, gauge invariant equations of motion. These gauge
symmetries guarantee that only physical DoF
propagate. Massive higher spin systems take a different route. Again,
their equations of motion are second derivative order, but are no
longer gauge invariant. Instead, a set of integrability conditions of
the equations of motion imply constraints that are required for  propagation
of only physical DoF. The PM system's
route is an intermediate one: Their second derivative field equations enjoy both
gauge invariances and integrability conditions implying {\it gauge
  invariant} constraints. As mentioned, the beauty of our BGG detour complex
approach to these systems is that it automatically produces a system
of gauge invariant equations, for a minimal field content, that
includes both equations of motion and constraints.

The BGG and detour apparatus is also linked to other approaches in the physics literature such as the higher spin unfolding  programme~\cite{Shaynkman:2004vu} whose $\sigma_-$ cohomology, see for example~\cite{Vasiliev:2009ck,Skvortsov:2009nv}, 
 is really the homology of the Kostant differential~\cite{Kostant}. Also, BRST machinery (which is intimately related to Lie algebra cohomology) applied to parabolic Lie superalgebras represented by differential operators acting on 
higher rank tensor bundles, has been used to construct detour operators for massless higher spin models and related systems describing supersymmetric black hole dynamics in~\cite{Cherney:2009mf,Cherney,Cherney:2010xh}.  

Despite its being a fundamental geometric structure, projective
geometry is still rarely utilized in physical settings, so we briefly
review its key ingredients, and those of metric
projective geometry, in Section~\ref{mProj}.  Our key tool for
handling projective geometries is the tractor calculus
of~\cite{BEG}. This is detailed in Section~\ref{tractors}. Detour
complexes and BGG sequences are powerful technologies, these are
described in generality in Section~\ref{dS}. There we also explain
their relationship to physical systems. Armed with all the above
machinery, we finally turn to explicit physical models in
Section~\ref{pmm}.  In Theorem~\ref{Bigtheorem}, using the BGG detour
complex apparatus, we establish the existence of a complex describing
the equations of motion and gauge invariances of a broad class of
models. Then in Theorem~\ref{contraintSys}, we show that the BGG
complex also encodes a gauge invariant system of constraints for those
models. In Section~\ref{s=3CC} we spell out the case of PM spin three on
constant curvature backgrounds. In Section~\ref{EBs} we study the
extension of our higher spin results to general Einstein
backgrounds. For spin two  there is no obstruction here, but already for
spin three, the BGG detour technology neatly characterizes an obstruction to
propagation in these spaces. There we also present a new complex coupling
spin three to a mixed symmetry field giving, at least, gauge invariant
dynamics in an Einstein background. Section~\ref{Action} discusses
action principles within the BGG formulation; here we focus on the
spin two  case. In Appendix~\ref{LRR}, we relate the log-radial reduction
technique (which is a popular method for studying higher spin systems
and their interactions) to projective tractor calculus.

\section{Metric projective differential geometry}\label{mProj}

At the level of underlying geometry we will exploit the 
interaction between metric and projective geometry. The discussion of
projective geometry here follows~\cite{{BEG,eastwood,GN}}
while~\cite{CGM,EM,GM} provide background theory for metric projective
theory.

\subsection{Projective Geometry}\label{proj}
Projective geometry is one of the simplest examples of a parabolic
geometry; it can be defined as a Cartan geometry modeled on the
homogeneous space~$Sl(n+1,\mathbb{R})/P$ where~$P$ is the parabolic
subgroup stabilizing a ray in~$\mathbb{R}^{n+1}$.  
A projective manifold is the structure $(M,\p)$ where $M$ is a smooth $n$-dimensional manifold and $\p$ is 
an equivalence class
of torsion-free affine connections, where ~$\widehat{\nabla}\sim \nabla$ if,
acting on any one-form~field $\omega$, or vector~field $v$, they are related by
\begin{equation}\label{cactus}
\widehat{\nabla}_a\omega_b=\nabla_a\omega_b-\Upsilon_a\omega_b-\Upsilon_b\omega_a
\ \Leftrightarrow\ 
\widehat{\nabla}_av^b=\nabla_av^b+\Upsilon_av^b+\delta_a^b\Upsilon_cv^c\, ,
\end{equation}
for some one-form~$\boldsymbol{\Upsilon}$. This definition is derived from the fact
that the relationship given in~\nn{cactus} is exactly the condition that
$\widehat{\nabla}$ and $ \nabla$ share the same geodesics
as unparametrised curves. This is a classical result; see~\cite{eastwood} for a modern treatment. Extending the above by linearity to a~$p$-form~$\omega_{bc\cdots d}$ we have,
$$
\widehat{\nabla}_a\omega_{bc\cdots d}=\nabla_a\omega_{bc\cdots d}-(p+1)\Upsilon_a\omega_{bc\cdots d}-(p+1)\Upsilon_{[a}\omega_{bc\cdots d]}\, ,
$$
and in particular, for a top form
$
\widehat{\nabla}^\faketop_a\omega^{\scalebox{.6}{\textrm{top}}}_{bc\cdots d}=\nabla_a^\faketop\omega^{\scalebox{.6}{\textrm{top}}}_{bc\cdots d}-(n+1)\Upsilon_a^\faketop\omega^{\scalebox{.6}{\textrm{top}}}_{bc\cdots d}
$.
Taking powers of the volume density bundle  gives the projective density bundle $$\mathcal{E}(w):=\big((\Wedge^nT^*M)^2\big)^{-w/2(n+1)}\, ;$$ 
and thus for a section~$\sigma$ of this we have:
$$
\widehat{\nabla}_a\sigma=\nabla_a\sigma+w\Upsilon_a\sigma\, .
$$ 

As a point of notation, in the following, given any vector bundle
$\cB$, we will write $\cB(w)$ as a shorthand for $\cB\otimes \ce(w)$,
and we say the vector bundle $\cB(w)$ (and any section thereof) has
{\em projective weight} $w$. Quite generally, we also use the same notation for a bundle and its section space.

Note that, at this juncture, there is no notion of a Riemannian
metric, nevertheless given a torsion-free affine connection  $\nabla$, its curvature tensor is given by
$$
[\nabla_a,\nabla_b]\omega_c=-R_{ab\,\,\,\,c}^{\,\,\,\,\,\,d}\omega_d\, ,
$$ 
and $R_{[ab}{}^d{}_{c]}=0$. Moreover, we can decompose the
curvature into its trace-free and trace pieces as~\cite{BEG}
\begin{equation}\label{dec1}
R_{ab}{}^d{}_c=W_{ab}{}^d{}_c+2\delta_{[a}^dQ^{\phantom{c}}_{b]c}-2Q^{\phantom{c}}_{[ab]}\delta^d_c\, ,
\end{equation}
where $W_{ab}{}^d{}_c$ and $Q_{ab}$  are the {\it projective Weyl} and {\it Schouten tensors}. The former obeys $W_{[ab}{}^d{}_{c]}=0=W_{ad}{}^d{}_c$ and is projectively invariant while the 
 projective Schouten tensor has no definite symmetry and transforms as
$$
\widehat{Q}_{ab}=Q_{ab}-\nabla_a\Upsilon_b+\Upsilon_b\Upsilon_a\, .
$$
It should not  be confused with its  conformal geometry counterpart denoted~$\P_{ab}$.
The curl of this, $\nabla_a Q_{bc}-\nabla_b Q_{ac}=: C_{abc}$,
defines the projective Cotton tensor.

\subsection{Scales} \label{scales}
Any $\nabla\in \p$ also gives a connection on any tensor bundle, and
in particular on the line bundle $(\Wedge^nT^*M)^2$. Conversely a
connection $\nabla \in \p$ is determined by a choice of connection on
$(\Wedge^nT^*M)^2$.  The skew part of the Schouten tensor, $Q_{[ab]}$,
is (up to a non-zero constant multiple) the curvature of $\nabla$ on
that line bundle.  As already used above, the line bundle
$(\Wedge^nT^*M)^2$ is
trivial. In fact it is also canonically oriented and the 
positive square root is the {\em volume density bundle}.  A nonvanishing section of
$(\Wedge^nT^*M)^2$, or equivalently any of its roots $\ce(w)$ with
$w\neq 0$, is called a {\em choice of scale} and in an obvious way
determines a line bundle connection preserving the given section
(which may be thought of as a global frame). Thus a choice of scale
determines a connection $\nabla\in \p$ which, by slight abuse of
terminology, we shall also call a choice of scale. (Such a connection
determines a section of $\ce(w)$, $w\neq 0$, up to multiplication by a
non-zero constant.) It is clear that for any such connection we have 
$$
Q_{[ab]}=0,
$$ 
and different choices of scale yield a transformation on the form
\nn{cactus} where $\boldsymbol{\Upsilon}$ is exact, {\it cf.}~\cite{CGM,GM,GN}. 

In our subsequent discussions we restrict to connections $\nabla\in
\p$ which correspond to a choice of scale; such connections form a
distinguished class so this results in no loss of generality.
Moreover, for simplicity, we will assume $M$ is orientable.

\subsection{Connecting with pseudo-Riemannian geometry}\label{mg}

Given a projective manifold~$(M,\p)$, a natural question is whether
there is a Levi-Civita connection in the projective class $\p$, and if
so, what the consequences are. A route towards  answering these
questions  is provided by the following result due to
Mikes~\cite{Mikes} and Sinjukov~\cite{Sinjukov}.
\begin{proposition}\label{metricproj}
If~$\widehat{\nabla}_a$  preserves some volume density, and there exists a metric tensor~$\sigma^{ab}\in \odot^2 TM(-2)$ satisfying the projectively invariant condition
\begin{equation}\label{tfsigma}
\textrm{\emph{trace-free }}(\widehat\nabla_a\sigma^{bc})=0\, ,
\end{equation}
then there is a projectively related connection~$\nabla$ which is
the Levi-Civita connection of the metric~$g^{ab}=\tau\sigma^{ab}$,
for some nonvanishing smooth density~$\tau\in\mathcal{E}(2)$, where ${\boldsymbol \Upsilon}$ (as defined
in~\eqref{cactus}) is given by~${ \boldsymbol \Upsilon}=-\frac12{\boldsymbol \nabla} \log \tau$.
\end{proposition}
The formulation here follows~\cite{EM}.

\begin{remark}
When the projective class of connections has a 
Levi-Civita connection $\nabla^g$,
the antisymmetric part of the projective Schouten tensor
vanishes. Then, in this scale,  
we can decompose the projective 
Weyl tensor into $SO$-irreducible parts, see~\cite{EM}:
$$ W_{ab}{}^d{}_c=\ring{W}_{ab}{}^d{}_c+\frac{2}{(n-1)(n-2)}\,
\delta_{[a}^d{}\ring R_{b]c}^{\phantom{c}} +\frac{2}{n-2}\, \ring
R_{[a}{}^dg_{b]c}^{\phantom{c}} \, ,$$ 
where~$\ring R_{bd}$ and $\ring W_{ab}{}^d{}_c$
 are the totally trace-free parts of, respectively, 
the Ricci and projective Weyl tensors. By construction  the latter is the usual,
conformally invariant, Weyl tensor~ $\W$ of Riemannian geometry (so
$\ring W=\W$). It follows immediately that, for an Einstein metric $g$, we have 
$$ 
W_{ab}{}^d{}_c=  \W_{ab}{}^d{}_c .
$$
\end{remark}

It turns out that many simplifications occur when dealing with an
Einstein metric. First we make a definition.
\begin{definition}
If there exists $\nabla^g\in {\p}$ where $\nabla^g$ is the Levi-Civita connection for a
metric, we call $(M,{\p})$ a {\it metric projective structure}. If in addition, $g$ is an Einstein metric, we say that this is {\it Einstein projective}. \end{definition}

When the structure is Einstein projective, computations performed
using the Einstein metric $g$ and its Levi-Civita
connection~$\nabla^g$ will be referred to as being done in an {\it
  Einstein scale}.

The following observation of~\cite{GM} is an easy consequence of equation~\nn{dec1} 
and the subsequent discussion.
\begin{proposition}\label{projconf}
 When $g$ is an Einstein metric,  the conformal Schouten tensor $\P$ for~$g$ and the projective Schouten tensor~$Q$ for~$\nabla^g$ 
  obey
$$
\P_{ab}=\frac{\Lambda}{(n-1)(n-2)}g_{ab}=\frac12\,  Q_{ab}\, .
$$
Moreover the conformal Weyl curvature $\W$ of~$g$ equals the projective Weyl curvature~$W$ of~$\nabla^g$.
\end{proposition}

\begin{remark}
Note that the scalar curvature is $\frac{2n\Lambda}{n-2}$ and $\Lambda$ in the above has been defined to coincide with the usual cosmological constant appearing in physics applications for which the Einstein tensor obeys $G_{ab}+\Lambda g_{ab}=0$. 
\end{remark}

\section{Tractor calculus for projective geometries}
\label{tractors}

On a general projective manifold $(M,\p)$ there is no distinguished
({\it  i.e.}, canonical) connection on the tangent bundle.  However the
projective structure $\p$ does determine a distinguished connection on
a related vector bundle of rank $(n+1)$ called the standard tractor
bundle; this connection is known as the (standard) projective tractor
connection and is due to 
Thomas~\cite{T}. The
modern treatment was initiated in~\cite{BEG}, and is equivalent to the
projective Cartan connection of 
Cartan~\cite{C}, see~\cite{CGtams}. Other developments relevant to our treatment can
be found in~\cite{CGM,GM}.

The {\em (projective) tractor bundle} $\cT$, or $\cT^A$ in an abstract
index notation,  can be
defined as follows. For each choice of a connection in the projective
class we identify the tractor bundle with the direct sum
$$
\mathcal{T}^A\cong TM(-1)\oplus
\mathcal{E}(-1)\ni\begin{pmatrix} v^a\\
\rho\end{pmatrix}\, ,
$$
where on the right we indicate how elements will be denoted. Equivalently for its dual, namely the {\em (projective) cotractor bundle}~$\mathcal{T}^*$, or $\cT_A$, we have
$$ \mathcal{T}_A\cong T^*M(1)\oplus \mathcal{E}(1)\ni\begin{pmatrix}
w_a & \sigma\end{pmatrix}\, .
$$
Changing the connection in the projective class according to \nn{cactus}, 
these transform as
\begin{equation}\label{ttrans}
\begin{pmatrix}v^a \\[1mm]
 \rho\end{pmatrix}\to 
 \begin{pmatrix}
v^a\\[1mm] \rho+\Upsilon_a v^a  
\end{pmatrix}\, ,
\qquad
\begin{pmatrix} \omega_a & \sigma\\\end{pmatrix}\to 
\begin{pmatrix}
  \omega_a + \Upsilon_a \sigma & \sigma
\end{pmatrix}\, ,
\end{equation}
and it is easily verified that this leads to well defined vector
 bundles on $(M,\p)$. The transformations \eqref{ttrans} mean that 
these bundles are filtered with composition series, respectively,
\begin{equation}\label{comp}
\cT = TM(-1)\flplus \ce(-1), \quad \mbox{and} \quad 
\cT^*= T^*M(1) \frplus \ce(1).
\end{equation}
For the second of these, for example, this means that there is a
canonical (so projectively invariant) short exact sequence of bundles
\begin{equation}\label{composition}
0\to \ce_a(1) \stackrel{Z^a_A}{\longrightarrow} \cT_A \stackrel{X^A}{\longrightarrow} \ce(1)\to 0, 
\end{equation}
where we have written $\ce_a$ as an abstract index notation for
$T^*M$. The canonical homomorphism $X^A$ may be viewed as a section of
$\cT^A(1)$ and is often called the {\em canonical tractor}. This also
gives the canonical bundle inclusion $X^A:\ce(-1)\to \cT^A$ indicated by the
dual composition the series for $\cT$.

In a given scale $\nabla\in \p$, we may define a covariant derivative on $\cT$ (and
its dual on~$\cT^*$) as follows:
\begin{equation}\label{gradients} 
 \nabla_a^\mathcal{T} \begin{pmatrix} v^b\\[1mm]\rho \\\end{pmatrix}=
 \begin{pmatrix}\nabla_a v^b+\delta_{a}^b\sigma
 \\[1mm]
 \nabla_a\rho-Q_{ab}v^b
 \end{pmatrix}\, ,\qquad
\nabla_a^\mathcal{T} \begin{pmatrix}  \omega_b & \sigma \\\end{pmatrix}=
 \begin{pmatrix} \nabla_a \omega_b+Q_{ab}\sigma& \nabla_a\sigma-\omega_a\end{pmatrix}\, .
\end{equation}
 Upon changing to a different connection in $\p$, as in \nn{cactus},
   it is easily verified that the right-hand-sides here transform
   according to \nn{ttrans}, signaling that $\nabla^\cT$ descends to
   a projectively invariant connection on $\cT$. This is the {\em
     projective tractor connection} and we denote it also by
   $\nd^\cT$.

The tractor connection determines a connection on all tensor products
of $\cT$ and its dual, in an abstract index notation such a bundle may
be denoted $\mathcal{T}_{A_1...A_p}{}^{B_1...B_q}$.  The connection above also induces connections on
``$Sl(n+1)$-irreducible'' tensor parts of these bundles. An
alternative perspective on this is via a principal bundle picture as
follows.

It is straightforward to construct an adapted frame bundle $\cG$ for
$\cT$, with frame transformations that respect the filtration structure
\nn{comp} and give $\cG$ a typical fibre isomorphic to the parabolic
$P$. The tractor connection then determines a Cartan connection $\om$
on~$\cG$, and this is the normal Cartan connection  for  projective geometry, see~\cite{CGtams}. We do not need the details of this here, but the point is that the tractor bundle and connection may then be viewed as induced by the Cartan connection through the standard representation of~$P$ on~$\mathbb{R}^{n+1}$ with 
$$
\cT=\cG\times_P \mathbb{R}^{n+1}.
$$ From the properties of Cartan connections it follows at once that
the Cartan connection similarly gives a connection on any associated
bundle 
\begin{equation}\label{assoc}
\cG\times_P \mathbb{V} , 
\end{equation}
where $\mathbb{V} $ is an
$Sl(n+1)$-representation space viewed as a $P$-representation by
restriction. This perspective is developed fully, in the setting of
general parabolic geometries, in~\cite{CGtams}. Here we denote any such {\em tractor connection} by $\nd^\cT$.

\subsection{Projective curvature} \label{pc}
The curvature of the tractor connection is given by
$$\Omega_{ab}{}^C{}_D=\left(\begin{array}{cc} 
W_{ab}{}^c{}_d
&0\\[3mm]C_{bad}&0\end{array}\right)\, ,
$$
where~$W_{ab}{}^c{}_{d}$ and $C_{abc}$ are  the projective Weyl and Cotton tensors.
 This is the invariant curvature 
 associated with projective geometry and a projective structure is flat if and only if $\Omega_{ab}{}^C{}_D=0$.

\subsection{The Thomas $D$-operator} \label{TD}

Given any tractor bundle $\cV$ there is a projectively invariant operator
$$
D: \cV(w)\to \cT^*\otimes \cV(w-1)
$$
known as the Thomas $D$-operator, see~\cite{BEG}. In a choice of scale    this is given  by 
\begin{equation}\label{ThomasD}
V\mapsto D_A V = \left(\begin{array}{c}
w V\\[1mm]
\nabla_a V
\end{array}\right)\, ,
\end{equation} where $\nabla_a $ is a coupling on the tractor connection with the
scale connection on $\ce(w)$, and the indices of $\cV$ are omitted.
It is an elementary exercise to verify that this transforms according
to \nn{ttrans}, and so $D$ is projectively invariant.

\subsection{Einstein projective structures}\label{EPS}

Recall that a metric is said to be {\em Einstein} if $R_{ab}=\Lambda
g_{ab}$ for some constant $\Lambda$. The Einstein condition has a
striking interpretation in projective geometry, and one that plays a
key {\it r\^ole} in our developments below. The following result is
established from different perspectives in~\cite{CGM} and
\cite{GM}. (The non-degenerate case  was first due to Armstrong
\cite{A}.)
\begin{proposition}\label{pM}
Let $(M,\p)$ be a projective manifold. There is an Einstein metric $g$
with Levi-Civita connection $\p$ if and only if there is a symmetric tractor field
$$
H^{AB}\in \odot^2\cT
$$ that is parallel for the tractor connection and of rank at least
$n$. If $H^{AB}$ is non-degenerate, then it is equivalent to a
non-Ricci flat metric $g$. If $H^{AB}$ is degenerate (of rank $n$)
then it is equivalent to Ricci flat metric.
\end{proposition}

In the case that $H^{AB}$ is non-degenerate there is an easy and conceptual way to
reconstruct the metric $g$ from $H^{AB}$, as follows. Denote by 
$H_{AB}$ its inverse. Then this determines a scale 
\begin{equation}\label{keysc}
\tau:= H_{AB}X^AX^B\in \ce(2)\, ,
\end{equation}
where $X^A$ is the canonical tractor from \nn{composition}. This or
its negative is a positive section (we assume $M$ connected) and so
may be used to trivialise density bundles. Thus from the sequence
\nn{composition} we may view $T^*M$ as a subbundle on $\cT^*$ and so
we obtain a metric on $T^*M$ by the restriction of $H$ to $T^*M$. This
is $g^{-1}$. It is a straightforward  use of the formula for the tractor connection to verify that the metric $g$ has its Levi-Civita connection in $\p$.

 In fact, a slight variant of this construction recovers $g$ in the
 case that $H^{AB}$ has rank~$n$, as shown in~\cite{CGM} and
~\cite{GM}, but this is less obvious. Some insight is gained by recalling 
from
Proposition~\ref{metricproj} that if we have a connection  and
tensor~$\sigma^{ab}$ in $\odot^2TM(-2)$ satisfying~\eqref{tfsigma},
then we have a metric connection in the projective class.  In fact
this equation is projectively invariant  and the operator on the left-hand-side is a first BGG operator.
In this case the BGG splitting operator~$\mathcal{L}_0$ applied to $\si^{bc}$ 
is given explicitly by 
$$
\mathcal{L}_0(\sigma^{bc}) =\begin{pmatrix} 
\sigma^{ab}&\frac{1}{n+1}\nabla_c\sigma^{cb}\\[2mm]
\frac{1}{n+1}\nabla_c\sigma^{ca}
&\frac{1}{n(n+1)}\nabla_c\nabla_d\sigma^{cd}+Q_{cd}\sigma^{cd}\end{pmatrix} .
$$
 If $\si^{ab}$ is
non-degenerate is clear that this has rank at least $n$ and $\si^{ab}$
is said to be a {\em normal solution} to \nn{tfsigma} if and only if
$H^{AB}:=\mathcal{L}_0(\sigma^{bc})$ is parallel~\cite{CGM}. By taking its
determinant in an obvious way, the solution $\si^{ab}$ determines a
scale $\tau\in \ce(2)$ and $\tau \si^{ab}=g^{ab}$ is the inverse of
the Einstein metric. Further details from this perspective may be found in~\cite{Cap:2013dva,CGM}.

Computed in the Einstein scale~$\tau$, the previous display  reduces to
\begin{equation}\label{tractormetric}
H^{AB}=\begin{pmatrix} 
g^{ab}&0\\[2mm]
0&\frac{1}{(n-1)}\Lambda 
\end{pmatrix}\, .
\end{equation}

Note that on an Einstein manifold the projective Cotton tensor
$\nabla_{[a} Q_{b]c}$ is zero. Moreover $\nabla^a W_{ab}{}^c{}_d=0$ by
dint of the contracted Bianchi identity. From these observations it
follows easily that 
$$
\nabla^a\Omega_{ab}{}^C{}_D=0\, ,
$$
and so the tractor connection is Yang-Mills.

\section{BGG and detour complexes}\label{dS}

Algebraic BGG resolutions appear
naturally in the representation theory of semi-simple Lie algebras, and
sequences of Verma modules associated to representations of Borel
subalgebras~\cite{BGG}. An extension to parabolic representation theory was
developed with Verma modules  replaced by generalised Verma modules
where the {\it r\^ole} of Borel subalgebras is replaced by parabolic
subalgebras~\cite{Lepowsky}. 

These constructions are, in a suitable sense, dual to complexes of
invariant differential operators on the corresponding $G/P$, where $G$
is a semi-simple Lie group and $P$ a parabolic subgroup. Such {\em
  homogeneous parabolic geometries} $G/P$ are the flat models for
parabolic geometries and the complexes so obtained are called BGG
complexes.  There are canonical curved analogues of these sequences
due to Baston, {\v{C}}ap et al, and others
\cite{Baston,Calderbank-Diemer,East-Rice}, but in general these sequences,
do not form complexes. 

\subsection{Gauge theory and differential complexes}\label{gaugeS} 
It has been known for some time that BGG complexes are 
related to gauge theories in physics, see {\it e.g.}~\cite{DiemerThesis}.
The latter are described in terms of  gauge
fields, curvatures and their Bianchi identities. Their kinematical
gauge structure is captured by a complex:
$$
\cdots\xrightarrow{\,\,\,\,\,\,}
\begin{array}{c}\textrm{gauge}\\  \textrm{parameters}\end{array}
\xrightarrow{\,\,\,\,\,\,} 
\begin{array}{c}\textrm{gauge}\\  
\textrm{potentials}\end{array}  \xrightarrow{\,\,\,\,\,\,} \textrm{curvatures} \xrightarrow{\,\,\,\,\,\,} \begin{array}{c}\textrm{Bianchi}\\  \textrm{identities}\end{array}\xrightarrow{\,\,\,\,\,\,} \cdots
$$
However, for the dynamics, a different sort of complex is required. 

It was observed, first in the setting of even dimensional conformal
geometry, that as well as BGG complexes there are {\em detour
  complexes} that, in addition to part of the BGG sequence, also use conformally
invariant ``long operators''~\cite{BG1,BG2}.  See also~\cite{Gasqui,Gilkey}, where a non-conformal but related construction is
developed and studied from a very different perspective.  Part of the
interest in these detour complexes stems from the fact that they can
be complexes in curved settings where the BGG sequences fail to be a
complex. This idea was further extended in~\cite{GSS}, where  classes
of detour complexes are constructed for each solution of the Yang-Mills
equations; these are conformally invariant in dimension~4. Both
\cite{BG2} and~\cite{GSS} link to variational constructions. Quantisation of the latter was taken up in~\cite{Gover:2006ha}.

For our current purposes,  detour complexes are important because
they provide a tool for generating gauge theories with dynamics where the diagram we seek takes the form: 
$$
\cdots\xrightarrow{\,\,\,\,\,\,}\begin{array}{c}\textrm{gauge}\\  \textrm{parameters}\end{array}\xrightarrow{\,\,\,\,\,\,} \begin{array}{c}\textrm{gauge}\\  \textrm{potentials}\end{array}  \xrightarrow{\,\,\,\textrm{detour operator}\,\,\,}  \begin{array}{c}\textrm{equations}\\ \textrm{of}\\ \textrm{motion}\end{array}  \xrightarrow{\,\,\,\,\,\,} \begin{array}{c}\textrm{Noether}\\  \textrm{identities}\end{array}\xrightarrow{\,\,\,\,\,\,} \cdots
$$ 
It is this approach that we follow below. First we sketch the
construction of BGG sequences focusing on projective geometries.

\subsection{BGG sequences}\label{bggS}
In this section we mostly follow the notation and conventions of~\cite{eastwood-prolongation}. On a projective manifold $(M,\p)$ 
there 
is a BGG sequence for every irreducible representation of $Sl(n+1)$, 
$${\mathbb V}
=\Yvcentermath1{\tiny\overbrace{\Ylinethick1.5pt\yng(6,4,3,2,1)}_{\hspace{-0.6cm}\vspace{-0.8cm}\,\,\,\,\,\,\,\,\,\,\,\,\,\,\,\,}^{m+k-1}} 
$$ which we indicate schematically here with a Young diagram.  Let us
fix $\mathbb{ V}$. The corresponding tractor bundle is associated to
the Cartan bundle via this representation, as in~\eqref{assoc}.

The {\em first operator} in the BGG sequence 
is then determined by the tractor connection acting on~$\cV$. This is
understood in terms of the general tools developed in~\cite{CSS-BGG,Calderbank-Diemer,Capover}.  We sketch the key ideas. 

The operator ${\mathcal D}:\mathcal{B}^0 \to \mathcal{B}^1$ is
projectively invariant and acts between weighted irreducible tensor
bundles $\cB^0$ and $\cB^1$. Let us say this has order $k$. Then its symbol 
is 
obtained by a
  composition of the form
$$
\mathcal{B}^0\to (\odot^kT^*M)\otimes \mathcal{B}^0\to (\odot^kT^*M)\circledcirc \mathcal{B}^0 \cong \cB^1 \, ,
$$ where $\circledcirc$ denotes the Cartan product.  Ignoring the
projective weight, the bundle $\cB^0$ is associated to the  Cartan 
bundle by an irreducible $Sl(n)$ representation (extended trivially to a $P$-representation)
$$\mathcal{B}^0\simeq \Yvcentermath1 {\tiny\overbrace{\yng(4,3,2,1)}^{m}} .$$
The projective
weight of $\cB^0$ is determined easily from $\cV$, but this detail is
not important for this general discussion. 
In the above we used unbolded and bolded Young tableaux for~$Sl(n)$
and~$Sl(n+1)$ representations, respectively.

Before describing the construction of $\cD$,
we introduce some algebraic ingredients.
There is a grading operator~$h\in \mathfrak{sl}(n+1)$ which (identifying $Sl(n+1)$ with its standard linear representation) can be given in the form
$$
h=\frac1{n+1}\, \left(\begin{array}{cc} {\bf 1}_{n \times n}&0\\[2mm]
0& -n
\end{array}\right)\, .
$$
This induces a  decomposition of $\frak{g}:=\frak{sl}(n+1)$
 $${\frak g}=
{\frak g}_{-1}\oplus{\frak g}_{0}\oplus
{\frak g}_{+1}\, , $$
with~$[h,{\frak g}_{\ell}]=\ell\,  {\frak g}_{\ell}$.
This is a $|1|$-grading, so $[\frak{g}_i,\frak{g}_j]\subset \frak{g}_{i+j}$.
On any $\frak{g}$-irreducible representation~${\mathbb V}$, the grading element can be diagonalized thus  producing  a natural splitting into eigenspaces for the action of~$h$
$$
{\mathbb V}_0\oplus {\mathbb V}_1\oplus ...\oplus {\mathbb V}_N\, ,
$$
with~${\mathbb V}_0$ (or ${\mathbb V}_N$) corresponding 
to the eigenspace with the lowest (or highest) eigenvalue; moreover, for any eigenvector~$v$ 
with eigenvalue~$i$,
we have~$\g_j {\mathbb V}_i\subset
{\mathbb V}_{i+j}$.  
The parabolic subgroup $P$ has Lie algebra
$\frak{g}_0\oplus \frak{g}_1$, and so we thus obtain a natural
filtration on the corresponding projective tractor bundle
$$
\cV=\cB^0\flplus \cdots \, .
$$

We now introduce the Kostant codifferential
\begin{eqnarray}
\partialf: \Wedge^{p+1} \g_{+1}\otimes {\mathbb V}&\longrightarrow& 
\Wedge^p \g_{+1}\otimes {\mathbb V}\, ,
\end{eqnarray}
defined by 
$$
\partialf(Z_1\wedge\cdots\wedge Z_{p+1}\otimes v)=\sum_{i=1}^{p+1}Z_0\wedge\cdots\hat{Z}_i\cdots\wedge Z_{p+1}\otimes Z_i v\, ,
$$ when $\frak{g}_{+1}$ is abelian.  The Kostant codifferential is
$P$-equivariant and nilpotent (meaning $\partialf\circ \partialf =0$),
and its homology is denoted~${H}_p(\g_{+1},{\mathbb V})$. We use now
the fact that $T^*M$ can be canonically identified with 
$\cG\times_P \g_{+1}$, and so 
the Kostant codifferential induces a projectively invariant 
map of tractor valued differential
forms: \begin{eqnarray} \partialf: \Wedge^{p+1}{\mathcal
    V}&\longrightarrow& \Wedge^p {\mathcal V}\, .
\end{eqnarray}
We denote by~${H}_p ({\mathcal V})$ the corresponding holomology bundle at degree $p$ 
and note that ${H}_p ({\mathcal V})= \cG\times_P
{H}_p(\g_{+1},{\mathbb V})$.  
We denote by $\pi$ the natural bundle map~$\pi:\textrm{ker}(\!\partialf)\to
{H}_p({\mathcal V})$. We are finally ready to construct the splitting
operator and the BGG operators.

By construction the bundle $\cB^0={H}_0 ({\mathcal V})$ and
$\cB^1={H}_1 ({\mathcal V})$ and the BGG sequence continues in this
way. 
The BGG machinery constructs for each $p$ a projectively invariant differential operator called a {\em  splitting operator}
$\mathcal{L}_p:{H}_p({\mathcal V})\to
 \Wedge^p{\mathcal V}$. This 
 is characterised as follows. 
\begin{proposition}
The following conditions determine a splitting operator:
\begin{itemize}
\item$\partialf\mathcal{L}_p(\alpha)=0$\\
\item$ \pi \, \mathcal{L}_p(\alpha)=\alpha$\\
\item$ \partialf \df^\nabla \mathcal{L}_p(\alpha)=0$
\end{itemize}
for any section~$\alpha$ of~${H}_p({\mathcal V})$.
\end{proposition}
We then obtain, by construction, the invariant differential operators
$$
{\mathcal D}_{p}:=\pi\circ \df^\nabla\circ \mathcal{L}_p:{H}_p({\mathcal V})\to {H}_{p+1}({\mathcal V})\, ,$$
and these operators form the BGG sequence:
\begin{diagram}
\scalebox{.7}{${\mathcal V}^\bullet$}&\rTo^{\dfs^\nabla} &
\scalebox{.7}{$\Wedge^1({\mathcal V}^\bullet)$}&\rTo^{\dfs^\nabla}
&\cdots\\&&&&\\ \uTo_{\mathcal{L}_0} & &\dTo{\pi}\uTo_{\mathcal{L}_1
}&&\dTo{\pi}\uTo_{\mathcal{L}_2
}\\ \mathcal{B}^0&\rTo^{\,\,\,\,\,\,\,{\mathcal D}_0\,\,\,\,\,\,\,} &
\mathcal{B}^1&\rTo^{\,\,\,\,\,\,\,{\mathcal
    D}_1\,\,\,\,\,\,\,}&\cdots \end{diagram}
 When the geometry is
flat, in the Cartan sense, this sequence becomes a complex.

In this section we sketched the construction of the BGG sequences for
projective geometries; the construction is quite similar for general
parabolic geometries.

\subsection{Yang--Mills detour complexes}\label{YMdet}
On a pseudo-Riemannian manifold, a basic detour complex
arises from the de Rham complex and the Maxwell operator $\deltaf \df$, where~$\deltaf$ is the formal adjoint of the exterior derivative
$\df$. 
\begin{equation}\label{max}
0\xrightarrow{\,\,\, \dfs \,\,\,}\Wedge^0M\xrightarrow{\,\,\, \dfs \,\,\,}\cdots \xrightarrow{\,\,\, \dfs \,\,\,} \Wedge^p(M)\xrightarrow{\,\,\,\,\,\,\deltafs \dfs\,\,\,\,\,\,}    \Wedge^p(M) \xrightarrow{\,\,\, \deltafs \,\,\,} \cdots\xrightarrow{\,\,\, \deltafs \,\,\,}  \Wedge^0M\xrightarrow{\,\,\, \deltafs \,\,\,}0 .
\end{equation}
In particular if $n$ is even and $p=n/2-1$ then this complex is
conformally invariant (at least after the introduction of conformally
weighted bundles for the right hand side of the complex). This is the simplest case 
of the family of (in general higher order) complexes found in~\cite{BG1}. 

If we drop the requirement of conformal invariance then the complex
\nn{max} is available in both dimension parities and for all
$p$-forms. In particular, if $p=1$, we obtain the complex 
\begin{equation}\label{maxd}
0\to \Wedge^0(M)\stackrel{\dfs\  }{\longrightarrow}\Wedge^1(M) \stackrel{\ \deltafs \dfs\ \ }{\longrightarrow}\Wedge^1(M)\stackrel{\deltafs\  }{\longrightarrow} \Wedge^0(M) \to 0\, ,
\end{equation}
which, in the spirit of discussion above, encodes the gauge theory of classical
source-free electromagnestism.

Now consider twisting the Maxwell detour \nn{maxd} with a connection
$\nabla$ on some vector bundle $\cV$ and replacing the exterior and
interior derivatives by~$\df^{\nabla}$,~$\deltaf^{\nabla}$ twisted by the
connection on~$\cV$.
This is in general doomed to failure as curvature means that~$\deltaf^{\nabla}\circ\,  \df^{\nabla}\, \circ {\df^{\nabla}} $ is not zero. However, as shown in~\cite{GSS}, there is a useful
variant on this first idea as follows.

\newcommand{\C}{\mathcal{C}}
\newcommand{\DD}{\mathcal{D}}

Suppose that the connection has curvature $\cF$.  We define its
action on $\cV$-valued differential forms by
$$
\begin{array}{rccc}
{\rm End}(\cF^\sharp) :& \Wedge^1(\cV)&\longrightarrow&\Wedge^{\, 1}(\cV)
\\[1mm]&\rotatebox{90}{$\in$}&&\rotatebox{90}{$\in$}\\
&\Xi_a^{\,\,\,\C}&\longmapsto& \cF_{a\,\,\,\,\,\,\,\DD}^{\,\,\,\, b\C}\,\Xi_b^{\,\,\,\DD}
\end{array}
$$ 
where $\C$, $\DD$ are abstract indices for the bundle $\cV$. We will often  
simply write ${\rm End}(\cF^\sharp) \Xi_a{}^{\C}$ for image of this map acting 
on $\Xi_a{}^{\C}$.

We construct now the operator
$$
\begin{array}{rccc}
\boldsymbol{M}^\nabla :& \!\!\!\Wedge^1(\cV)&\longrightarrow&\!\!\!\!\Wedge^{\, 1}(\cV)\\[1mm]&\rotatebox{90}{$\in$}&&\!\!\!\rotatebox{90}{$\in$}\\
&\Xi_a^{\,\,\,C}&\longmapsto&\!\!\!\left(\deltaf^{\nabla}\! \df^{\nabla}-{\rm End}(\cF^\sharp) \right)\Xi_a^{\,\,\,C}\, ,
\end{array}
$$
which leads to the following theorem:
\begin{theorem}(See \cite{GSS}.)  \label{4.2}
The sequence of operators 
$$
0\longrightarrow  \Wedge^0(\cV) \xrightarrow{\,\,\, \dfs^{\nabla} \,\,\,} \Wedge^1(\cV)\xrightarrow{\,\,\,\,\,\deltafs^{\nabla} \dfs^{\nabla}-{\rm End}(\cF^\sharp)\,\,\,\,\,}\Wedge^1(\cV)\xrightarrow{\,\,\, \deltafs^{\nabla} \,\,\,}\Wedge^0(\cV)\longrightarrow 0\, ,
$$ is a complex if and only if the curvature~$\cF$ satisfies the Yang--Mills equation \begin{equation}\label{YM}\deltaf^{\nabla} \cF=0\, .\end{equation} If $\nabla$ preserves a
metric on $\cV$ then the complex is formally self-adjoint.
\end{theorem}
\begin{proof}
A short calculation shows that 
\begin{equation}\label{ymc}
\bm{\!M\!\!}^{\!\nabla}\!  \df^\nabla = \varepsilon (\!\deltaf^{\nabla} \cF) \quad \mbox{and} \quad
\deltaf^\nabla \! \bm{\!M\!\!}^{\!\nabla} = -\iota (\!\deltaf^{\nabla} \cF)  \, ,  
\end{equation}
where $\iota$ and $\varepsilon$ denote interior and exterior multiplication. Finally, in the case that $\nabla$ preserves a metric, 
the operator~${\!\bm{M}\!\!}^{\nabla}$ is formally self-adjoint by construction. Then, also by construction, the sequence is formally self-adjoint. 
\end{proof}
Following~\cite{GSS}, we call the sequence of the theorem
the {\em Yang-Mills detour complex} and say a connection is {\it Yang--Mills} when it obeys the source-free Yang-Mils equation~\eqref{YM}. There is such a complex for every
Yang-Mills connection,  but our interest here is that some of these
can be used to generate other interesting complexes. A first
observation in this direction is that if a connection on a vector
bundle $\cU$ is Yang-Mills, then so are the induced connections on the
dual of $\cU$, on tensor powers of these and on tensor parts
thereof. So in each case there is such a complex. We wish to combine
this observation with the next less obvious construction.

\subsection{Translating}\label{trans}
Consider the general situation
\begin{equation}\label{cdi}
\begin{diagram}
&\Wedge^0(\cV) &\rTo^{\dfs^\nabla} &\Wedge^1(\cV)&\rTo^{\Ms^\nabla} 
&\Wedge^1(\cV)&\rTo^{\,\,\,\deltafs^\nabla\,\,\,} &\Wedge^0(\cV) \\
&\uTo_{\mathcal{L}_0} & &\uTo_{\mathcal{L}_1 }&&\dTo^{\mathcal{L}^1}&&\dTo^{\mathcal{L}^0}\\
&E&\rTo^{\,\,\,\,\,\,\,\mathcal D\,\,\,\,\,\,\,} &F &\rTo^{\,\,\,\,\,\,\,M\,\,\,\,\,\,\,} &F^\star&\rTo^{\,\,\,\,\,\,\,\mathcal D^\star\,\,\,\,\,\,\,}&E^\star
\end{diagram}
\end{equation}
where $\cD:E\to F$ and $\cD^\star: F^\star\to E^\star$, $\cL_0:E\to
\Wedge^0(\cV)$, $\cL_1:F\to \Wedge^1(\cV)$, $\cL^1: \Wedge^1(\cV)\to
F^\star$, and $\cL^0:\Wedge^0(\cV)\to E^\star$ are differential operators and 
$M$ is defined to be the composition 
$$
\cL^1{\!{\bm M}\!\!}^\nabla \cL_1:F\to F^\star.
$$

Suppose now that the left and right squares are commutative:
$$
\df^{\!\nabla}\mathcal{L}_0=\mathcal{L}_1{\mathcal D}\,\,\,\,\,\textrm{and}\,\,\,\,\,{\mathcal D}^\star\mathcal{L}^1=\mathcal{L}^0\deltaf^\nabla\, .
$$
Then it follows that
$$
\begin{array}{rcl}
M{\mathcal D}&=&\mathcal{L}^1\epsilon(\deltaf^\nabla \Omega)\mathcal{L}_0\, ,\\[3mm]
{\mathcal D}^\star M&=&-\mathcal{L}^0\iota(\deltaf^\nabla \Omega)\mathcal{L}_1\, .
\end{array}
$$ 
Thus, if the connection is  Yang--Mills
(\emph{i.e.},~$\deltaf^\nabla \cF=0$), it follows at once from Theorem
~\ref{4.2} that the lower differential sequence 
$$
E \xrightarrow{\,\,\, {\mathcal D} \,\,\,}F \xrightarrow{\,\,\,\,\,M \,\,\,\,\,}F^\star \xrightarrow{\,\,\, {\mathcal D}^\star \,\,\,}E^\star\, ,
$$ is a complex. Furthermore, observe that if $E$ and $F$ are tensor
bundles and also the connection $\nabla$ preserves a metric on $\cV$,
then we can construct the last square by taking adjoints of all the
differential operators from the first square. In this case the
commutativity of the last square is immediate from commutativity of
the first square, and by using once more Theorem~\ref{4.2} it follows
that the lower sequence is formally self-adjoint.

In summary we have recovered the following paraphrasing of a result
from~\cite{GSS}:
\begin{theorem} \label{4.2+} Suppose that the connection $\nabla$ on $\cV$ is metric and Yang-Mills. Suppose also that the first square of the diagram \nn{cdi} commutes. Then we obtain  a formally self-adjoint detour complex 
\begin{equation}\label{detd}
 0\longrightarrow E \xrightarrow{\,\,\, {\mathcal D} \,\,\,}F
\xrightarrow{\,\,\,\,\,M \,\,\,\,\,}F^\star \xrightarrow{\,\,\, {\mathcal
    D}^\star \,\,\,}E^\star \longrightarrow 0 .
\end{equation}
\end{theorem}

In the spirit of the discussion in Section ~\ref{gaugeS},
we can use the construction leading to Theorem ~\ref{4.2+} as a tool
for generating equations on motion $M$ operators on potentials in $F$. Then, also by
construction, these are invariant with respect to the gauge
transformations $\cD: E\to F$, and satisfy Bianchi identities
$\cD^\star:F^\star\to E^\star$.

\subsection{BGG detour complexes} \label{BGG-detour}
We use the term {\em BGG detour complex} to mean detour complexes
which use, in part, operators from the BGG complexes. 

In the following we will be, in particular, interested in using Theorem
~\ref{4.2+} to generate BGG detour complexes of the form \nn{detd},
where $\cD$ is a first BGG operator. In these constructions, suitable
projective tractor bundles will play the {\it r\^ole} of $\cV$, and we recall
that, in the case the projective structure is Einstein and non
Ricci-flat, these have a metric preserved by the tractor
connection.

\section{Partially massless models of maximal depth}\label{pmm}

Here, on non-flat constant curvature backgrounds, we show, using
the Yang-Mills detour theory, that for any spin $k\in \mathbb{Z}_{\geq 2}$,
there is a PM  gauge invariant equation of motion and
gauge invariant constraint system for fields in $\odot^kT^*M$. The gauge operators 
are  given by the restriction to constant curvature manifolds  of 
the projectively invariant first BGG operators 
\begin{equation*}
{\mathcal D}: \ce(k-1) \to \odot^kT^*M(k-1)\, .   
\end{equation*}
These linear operators take the form 
$$
{\mathcal D}(\si)= \nabla_{(a}\cdots \nabla_{c)}\si + \mbox{ lower order terms} .
$$ General algorithms for the explicit formulae for these are
available in~\cite{GEsigma}. In fact these formulae are members of general
families that take the same form on all parabolic geometries~\cite{Cald-Sou,Go89}. The examples of order three and  two are given in, respectively, 
\nn{spin3proj} and \nn{1BGGspin2} below.

\begin{theorem}\label{Bigtheorem}
Let $k$ be a positive integer. 
On a constant curvature manifold the Yang-Mills detour complex
associated with the projective tractor connection determines a
canonical formally self-adjoint detour complex 
$$
0\to \ce \stackrel{\ {\mathcal D}\ }{\longrightarrow} \odot^kT^*M  \stackrel{\ M\ }{\longrightarrow} \odot^kT^*M \stackrel{\ {\mathcal D}^\star\ }{\longrightarrow} \ce \to 0 ,
$$ where ${\mathcal D}^\star$ is the (order $k$) adjoint of the operator ${\mathcal D}$ and
the equation of motion operator~$M$ is second order.
\end{theorem}
\begin{proof}
We work first on any projectively flat structure $(M,\boldsymbol{p})$.
 
There is a projectively invariant BGG splitting operator
$$
\cL_0:\ce(k-1)\to \odot^{k-1}\cT^* .
$$ Thus the composition $\df^{\nd} \cL_0:\ce(k-1)\to
\Wedge^1(\odot^{k-1}\cT^*) $ is also projectively invariant.  
Next, by the standard BGG theory
$H_1(\odot^{k-1}\cT^*)\cong \odot^kT^*M(k-1)$ 
(and we identify these spaces),
$\partialf \circ \df^{\nd}\circ  \cL_0  =0 $, and the first BGG operator above
is given  by $\pi\circ  \df^{\nd}\circ  \cL_0 $, where $\pi$ is the bundle 
map arising from the map from $\ker (\partialf)$ to homology.

In fact we can say more. First we observe that, from the composition
series for the tractor bundle, there is a projectively invariant
bundle inclusion
\begin{equation}\label{ins} \textstyle
\imath : \odot^kT^*M(k-1)\to \bw^1(\odot^{k-1}\cT^*),
\end{equation}
and so it follows from the characterising properties of $\cL_1$ that 
$\cL_1=\imath$. But then it follows that 
\begin{equation}\label{comm}
 \df^{\nd}\circ  \cL_0 = \cL_1 \circ {\mathcal D} .
\end{equation}
This last result holds because, from the classification of
 projectively invariant differential operators on the sphere, there is
 only one, up to multiplication by a non-zero constant, projectively
 invariant operator on $\ce(k-1)$, and this is the operator ${\mathcal D}$. On the
 other hand, the bundle $\odot^kT^*M(k-1)$ occurs only once in the
 composition series for $\bw^1(\odot^{k-1}\cT^*)$, and this is realised by
 $\imath$.

This establishes a commuting first square, as in the diagram \nn{cdi}
leading to Theorem~\ref{4.2+}. Now we restrict to a constant curvature
manifold $(M,g)$, with $\Lambda\neq 0$. This has the projective
structure $\boldsymbol{p}=[\nabla^g]$ and we note that $\p$ is, in
particular, projectively flat, so the above results are
available. Furthermore, in this setting the manifold is Einstein so we
may use $\tau$ (see~\eqref{keysc})  to trivialise density bundles and there is a parallel
tractor metric given by~\nn{tractormetric}. This metric induces a metric on
$\odot^{k-1}\cT^*$. Thus we may write adjoints for all of the operators in the
first square and so obtain a commuting last square, as in the diagram~\nn{cdi}. Furthermore using this metric, the operator $\boldsymbol{M}^\nabla$ is
formally self adjoint and so the claimed system follows from Theorem~\ref{4.2+}.
\end{proof}

By construction then, the operator $M:\odot^kT^*M \to \odot^kT^*M $ is gauge
invariant with respect to the order $k$ gauge operator ${\mathcal D}$, while
${\mathcal D}^\star: \odot^kT^*M\to \ce $ provides the integrability conditions
({\it i.e.} ``Bianchi identities'') for this.

\begin{remark} 
For those familiar with tractor calculus, there is an 
alternative approach to  aspects of the proof above.  For example,
in the projectively flat setting (as in the theorem) it is easily
verified that the splitting operator $\cL_0:\ce(k-1)\to \odot^{k-1}\cT^*$, is
given explicitly by
 $$
 \si\mapsto D_{A_1}\cdots D_{A_{k-1}}\si\, ,
$$
  where $D_A$ is the projective Thomas $D$-operator.  The
  right-hand-side here is symmetric since on projectively flat
  manifolds the $D$-operators mutually commute. Then
  $$X^{A_1}D_{A_1}\cdots D_{A_{k-1}}D_{A_k}\si =0\, , $$ 
  and from this it
    follows  that \nn{comm} holds.
\end{remark}

\subsection{Constraints} In our current context, the detour 
construction leading to Theorem ~\ref{4.2+} encodes more than the
detour complex of Theorem ~\ref{Bigtheorem}. As well as  the Bianchi
identity on the equation of motion operator $M$ of the Yang-Mills detour complex, it also
captures gauge invariant constraints on the potential.
\begin{theorem} \label{contraintSys}
Applied to give equations on the potential, the differential operator 
\begin{equation}\label{coneq}
\textstyle 
\boldsymbol{M}^\nabla\!\circ \cL_1:\odot^kT^*M\to \bw^1 (\odot^k\cT^*), 
\end{equation}
extends the equation of motion of the BGG detour operator $M: \odot^kT^*M
\stackrel{\ M\ }{\longrightarrow} \odot^kT^*M$ (of Theorem ~\ref{Bigtheorem})
by a system of gauge invariant constraints.

There are relations between these constraints and the equations of
motion captured by the fact that image of $\boldsymbol{M}^\nabla\circ
\cL_1$ lies in the kernel of $\deltaf^\nabla$.
\end{theorem}
\begin{proof}
The first statement follows at once from the fact that 
$$
\boldsymbol{M}^\nabla\circ {\mathcal L}_1\circ {\mathcal D}=0 ,
$$ since the diagram \nn{cdi} is commutative. The last statement is
adjoint of this. 
\end{proof}

We are now ready to take up examples. There is no strictly partially
massless system for spin one. We will treat spin two in more general
setting below, so we begin with spin three.

\subsection{Spin three}\label{s=3CC}

Theorem~\ref{Bigtheorem} shows that the BGG detour complex
produces a system of gauge invariant equations of motion and constraints for higher spin fields in non-flat, constant curvature backgrounds. These spaces arise from the intersection of projectively flat and Einstein projective structures.
In the following we show, as a special case, that this recovers 
 the standard, maximal depth,  PM spin three theory in four dimensions of~\cite{Deser:2001pe,Deser:2001us}.
This 
 theory is described in terms of a totally symmetric rank three  tensor~$\varphi_{abc}$ and a scalar auxiliary~$\chi$, introduced
in order that the equations of motion are Lagrangian. These read
\begin{equation}\label{bigformaggio}
\left\{
\begin{array}{l}
\Delta\varphi_{abc} -\frac {5\Lambda} 3 \varphi_{abc}-3\nabla_{(a}\nabla.\varphi_{bc)}+3\nabla_{(a}\nabla_{b}\bar\varphi_{c)}-3g_{(ab}\Delta\bar\varphi_{c)}\\[1.5mm]
\phantom{\Delta\varphi_{abc}}
+3g_{(ab}\nabla^{d}\nabla^{e}\varphi_{c)de}\, -\, \frac 3 2 g_{(ab}\nabla_{c)}\nabla. \bar\varphi\, +\, \frac 3 4 g_{(ab}\nabla_{c)}\chi=0\, ,\\[4mm]
\Delta \chi+\frac{\Lambda}3\nabla.\bar\varphi=0\, ,
\end{array}
\right.
\end{equation}
where $\bar\varphi_a:=\varphi_{ab}{}^b$ and $\nabla.\varphi_{bc}:=\nabla_a\varphi^a{}_{bc}$.
These equations imply, as integrability conditions, the constraints
\begin{equation}
\label{c1}
\left\{
\begin{array}{l}
\nabla_{(a}\nabla_{b)_\circ}\chi-\Lambda\nabla.\varphi_{(ab)_\circ}+2\Lambda\nabla_{(a}\bar\varphi_{b)_\circ}=0\, ,\\[3mm]
\nabla_a\chi+\frac \Lambda 3 \bar\varphi_a=0\, .
\end{array}
\right.
\end{equation}
We use the notation $(\cdots)_\circ$ to indicate the trace-free, symmetrized part of a group of indices. 
In addition to these constraints, this system enjoys a higher derivative, scalar, gauge invariance 
\begin{equation*}
\left\{\begin{array}{l}
\delta\varphi_{abc}=\left(\nabla_{(a}\nabla_{(b}\nabla_{c)_\circ\!)}+\frac \Lambda 2\,  g_{(ab}\nabla_{c)}\right)\varsigma\, ,\\[3.5mm]
\delta\chi=-\frac \Lambda 3\!\left(\Delta+\frac{10\Lambda}{3}\right)\varsigma\, .
\end{array}\right.
\end{equation*}

In the above, the key ingredient linking the PM model to the BGG machinery is the gauge operator $\nabla_{(a}\nabla_{(b}\nabla_{c)_\circ\!)}+\frac \Lambda 2 g_{(ab}\nabla_{c)}$ acting on the scalar gauge parameter~$\varsigma$. For that, we must relate it to a projectively invariant operator.
To facilitate this we remove all instances of the inverse metric; this can be achieved by defining the trace-adjusted field 
\be\label{happy}
\psi_{abc}:=\varphi_{abc}+\frac 1 2 g_{(ab}\bar\varphi_{c)}\, , 
\ee
whose gauge transformation is
\be
\label{gtnew}
\delta\psi_{abc}=\left(\nabla_{(a}\nabla_b\nabla_{c)}+\frac{ 4 \Lambda}{3}g_{(ab}\nabla_{c)}\right)\varsigma\, .
\ee
Remarkably, the projectively invariant operator
\begin{equation}\label{spin3proj}
\begin{array}{ccccc}
{\mathcal D}\!\!\!&:&\!\!\!\mathcal{E}(2)&-\!\!\!\longrightarrow&\hspace{-1.2cm}\odot^3T^*M(2)\\[1mm]
&&
\!\!\rotatebox{90}{$\in$}&&
\hspace{-1.2cm}\rotatebox{90}{$\in$}\\
&&\!\!\sigma&\mapsto&\hspace{-4mm}
\scalebox{.91}{$\frac12 \left[\nabla_{(a}\nabla_{b}\nabla_{c)}+4Q_{(ab}\nabla_{c)}+2(\nabla_{(a}Q_{bc)})\right]$}\sigma\end{array}
\end{equation}
matches the gauge operator
appearing in~\eqref{gtnew}
when computed in the Einstein scale (for which~$Q_{ab}=\frac \Lambda 3 g_{ab}$). 

As a further happy consequence of the field redefinition~\eqref{happy}, the {\it divergence constraint} on the first line of~\eqref{c1}, using the second line of that display to eliminate $\chi$, simplifies~to 
\bes
\label{happyconsequence}
\nabla.\psi_{ab}=\nabla_{(a}\bar\psi_{b)}\, .
\ees

According to the BGG construction described in the previous section, the $Sl(5)$ module required for maximal depth, spin three  is given by the Young diagram~$\Yvcentermath1{\tiny \Ylinethick1.5pt \yng(2)}$, 
and hence sections of
 the symmetric product of the cotractor bundle, denoted by~$\mathcal{T}_{(AB)}$. In a given scale, a general section of this bundle splits as
$$
 \begin{pmatrix}w_{ab}&v_a\\ v_b&\tau\end{pmatrix}=:\begin{pmatrix} \tau\\ v_a\\ w_{ab}\end{pmatrix}\, ,
$$
on which the tractor connection acts as follows:
$$
\nabla^{\mathcal{T}}_c
\begin{pmatrix}\tau\\ v_a\\ w_{ab}\end{pmatrix}=\left(\begin{array}{c}\nabla_c \tau-2v_c\\[1mm] \nabla_c v_a+Q_{c a}\tau-w_{c a}\\[2mm]\nabla_c w_{ab}+2Q_{c (a}v_{b)}\end{array}\right)\, .
$$
The first part of the diagram~\eqref{cdi} (encoding  the gauge transformation)  is  easily computed:
$$
\begin{diagram}
\mathcal{T}_{(AB)}&\ni\scalebox{.7}{$\left(\begin{array}{c}\sigma\\[3mm]\frac 1 2 \nabla_a\sigma\\[3mm](\frac 1 2 \nabla_a\nabla_b+Q_{ab})\sigma \end{array}\right)$}&\rMapsto^{\, \dfs^\nabla} &\scalebox{.7}{$\left(\begin{array}{c}0\\[3mm]0 \\[3mm] \frac 1 2 {\bm \nabla\!} \nabla_a \nabla_b\sigma+{\bm Q\!\!}_{ (a}\nabla_{b)}\sigma+{\bm \nabla\!\!} (Q_{ab} \sigma)\end{array}\right)$}&\,\in\,& 
\Wedge^1(\mathcal{T}_{(AB)})&\,\ni\,&\scalebox{.8}{$\left(\begin{array}{c}0\\ 0\\{\bm \psi\!\!}_{ ab} \end{array}\right)$}
\\&&&&&&&\\
&
\uMapsto_{\mathcal{L}_0} & &\dMapsto{\pi}&&&&\uMapsto_{\mathcal{L}_1 }\\
\mathcal{E}(2)\ni
\hspace{-2.5cm}&\sigma&\rMapsto^{\hspace{-1.3cm}\mathcal D} &
\scalebox{.8}{$\left(\frac12\nabla_{(a}\nabla_b\nabla_{c)}+2Q_{(ab}\nabla_{c)}+(\nabla_{(a}Q_{bc)})\right)$}
\sigma
&\,\in\,& \odot^3T^*M(2)&\,\ni\,&\hspace{-1mm}\psi_{abc}
\end{diagram}
$$
Here ${\!\bm \psi\!\!}_{ab}$  and ${\!\bm Q\!\!}_a$ denote the one-forms made from $\psi_{abc}$ and $Q_{ab}$ by soldering (the soldering form is denoted ${\!\bm e\!\!}^a$) and $\nabla$ is any connection in~$\p$.

We now calculate in the Einstein scale and use the parallel tractor metric:
$$
H_{AB}=\begin{pmatrix}g_{ab}&0\\[1mm]
0& \frac{3}{\Lambda} \end{pmatrix}\, .
$$
This allows us to 
 construct the detour long operator~$\boldsymbol{M}^\nabla\equiv \deltaf^\nabla \df ^{\nabla}-{\rm End}(\Omega^\#)$:
$$
\begin{diagram}
 \Wedge^1(\mathcal{T}_{(AB)})&\,\ni\,&\scalebox{.7}{$\left(\begin{array}{c}0\\ 0\\{\!\bm \psi\!\!}_{ ab} \end{array}\right)$}&\rMapsto^{\ {\Ms}^\nabla\ }& \scalebox{.7}{$\left(\begin{array}{c}0\\ \nabla . {\!\bm \psi\!\!}_{ a}-{\bm \nabla\!} \bar\psi_a \\[1mm]\Delta{\!\bm\psi\!\!}_{ab}-{\bm \nabla\!}\nabla.\psi_{ab}+\frac {2\Lambda} 3 {\!\bm e\!\!}_{(a}\bar\psi_{b)}-\frac {5\Lambda} 3 {\!\bm\psi\!\!}_{ a b} \end{array}\right)$}& \in \Wedge^1(\mathcal{T}_{(AB)})
 \\ &&&&&\\
&&\uMapsto_{\mathcal{L}_1 }&&\dMapsto_{\mathcal{L}^1 }&\\
 \odot^3T^*M&\,\ni\,&\scalebox{.9}{$\psi_{abc}$}&\rMapsto^{\!\!M}&\scalebox{.7}{$ \Delta\psi_{ abc}-\nabla_{(a}\nabla.\psi_{bc)}+\frac{ 2\Lambda} 3 g_{( ab}\bar\psi_{c)}-\frac {5\Lambda} 3 \psi_{ a bc}$}&\hspace{-5.5mm}\in \odot^3T^*M
\end{diagram}
$$
The operator $M$ on the bottom line of this diagram should be compared with the equations of motion of the PM system. By construction, it annihilates the gauge operator of~\eqref{gtnew}. Before making this comparison, we first note that the detour complex on the top line of the diagram encodes further information.
First, the second line in the image of $\boldsymbol{M}^\nabla$ reads
(making the form index explicit)
$$
\nabla.\psi_{ab}-\nabla_a \bar\psi_b=0\, .
$$
This equation has both a symmetric and antisymmetric piece.
The latter implies that~$\bar\psi_a$ is a closed one-form, and thus locally the gradient of some scalar, which we identify with the auxiliary field~$\chi$:
$$
\bar\psi_a=-\frac6\Lambda\nabla_a \chi\, .
$$
This reproduces the second constraint in~\eqref{c1}. Turning to the
totally symmetric piece, it immediately reproduces the divergence constraint as expressed in~\eqref{happyconsequence}.

It remains to analyse the bottom slot of the image of $\boldsymbol{M}^\nabla$. 
It is not difficult to verify that its mixed symmetry part vanishes, modulo the divergence constraint~\eqref{happyconsequence}; here one also uses that the structure is Weyl-flat. 
The totally symmetric part of the bottom slot in the image of~$\boldsymbol{M}^\nabla$  matches exactly the equation of motion~\eqref{bigformaggio} for~$\varphi_{abc}$, upon employing~\eqref{happy}, again modulo~\eqref{happyconsequence}. Since the operator ${\mathcal L}^1$ projects the bottom slot onto its totally symmetric part, this establishes our claim.

Finally, we note, that for any projectively flat structure, 
$\df^\nabla {\mathcal L}_1(\varphi_{abc})$ produces a tractor for which only the bottom slot, $\boldsymbol{\nabla} \boldsymbol{\varphi}_{bc}$, is nonvanishing, and hence both projectively and gauge invariant. This quantity is the spin three generalization of the PM curvature found in~\cite{DeserEM}.

\section{Einstein Backgrounds}\label{EBs}

Einstein metrics play a special {\it r\^ole}
in projective geometry. In particular, recall from Section~\ref{EPS}, that a metric
$g$ which is Einstein but not Ricci flat, is equivalent to 
 a non-degenerate parallel metric~$H$ on the projective tractor bundle~$\cT$~\cite{CGM}.
We  now study projective BGG sequences and possible detour complexes in an Einstein setting. This allows us to address the physical question of higher spin propagation on Einstein manifolds.

\renewcommand{\B}{{\mbox{\sf B}}}

\subsection{Partially massless spin two}
We begin our Einstein investigation with the propitious,  spin two case.
As mentioned in the introduction,  PM spin two  was first discovered in studies of constant curvature, lightcone propagation~\cite{Deser:1983tm} (see as well~\cite{Deser:2001pe,Deser:2001xr}). However, it also has an interesting geometric, conformal gravity, origin~\cite{Maldacena,Deser:2013bs,Deser:2012qg}:
 Consider the four-dimensional conformal gravity action,
$$
S= \int  \epsilon_{abcd} {\! \bm W\!\!}^{ab}\wedge {\!\bm W\!\!}^{cd}\, . 
$$
This action is extremized by
 vanishing of the Bach tensor~$\B_{cd}$. The latter can be defined
 by a differential operator acting on the Schouten tensor 
  $$
\B_{ab}=\Big(\!-\delta_{a}^c\, \Delta\,  \delta_{b}^d+\nabla^c\delta^d_{(a}\nabla_{b)}   +\W^{c}{}_{ab}{}^d\Big)\P_{cd}\, .
$$ 
Defining~$\varphi_{ab}=\frac{\Lambda}{6}g_{ab}-\P_{ab}$, then the
Bach flat condition of the above display, dropping terms nonlinear in
$\varphi$, becomes (denoting the trace by $\bar\varphi:=\varphi^a{}_a$
and $\nabla.\varphi_a:=\nabla^b\varphi_{ab}$)
\begin{equation}\label{PM2eom}
\begin{split}
\Delta\varphi_{ab}-2\nabla_{(b}\nabla.\varphi_{a)}&+g_{ab}\nabla.\nabla.\varphi+\nabla_a\nabla_b\bar\varphi-g_{ab}\Delta\bar\varphi\\[2mm]&-2\W_{a}{}^{cd}{}_b\varphi_{cd}-\frac 4 3\Lambda(\varphi_{ab}-\frac 1 4 g_{ab}\bar\varphi)=0\, .
\end{split}
\end{equation}
The above display is exactly the PM equation of motion in an Einstein background~\cite{Dolan:2001ih} while the quantity $\varphi_{ab}$ measures the failure of $g_{ab}$ to be Einstein.
The divergence constraint
\be\label{spin2constraint}
\nabla^b\varphi_{ab}=\nabla_a\bar\varphi\, ,
\ee
 follows as an integrability condition of~\eqref{PM2eom}. Moreover, 
the PM equation of motion is variational and enjoys a higher derivative gauge invariance 
\be\label{spin2gaugeinvariance}
\delta\varphi_{ab}=\big(\nabla_a\nabla_b+\frac \Lambda 3 g_{ab}\big)\varepsilon\, .
\ee
The gauge operator here is  intimately related to projective geometry and BGG sequences.

\subsection{Spin two BGG}

We now return to projective geometry and
consider   the operator 
\be\label{1BGGspin2}
\begin{array}{rccc}
{\mathcal D} :& \mathcal{E}(1)&\longrightarrow&\!\!\!\!\odot^2 T^*M(1)
\\[1mm]&\rotatebox{90}{$\in$}&&\rotatebox{90}{$\in$}\ \ \ \\
&\sigma&\mapsto&\!\!\!\left(\nabla_{(a}\nabla_{b)}+Q_{(ab)}\right)\sigma
\end{array}
\ee
For Einstein projective structures,  computing in the Einstein scale, the above  reproduces the gauge operator appearing in~\ref{spin2gaugeinvariance}.

Using the parallel tractor metric
we construct a detour complex which describes  the spin two  PM system on Einstein backgrounds. 
The first
 BGG operator is the one displayed  in~\eqref{1BGGspin2}. 
We have summarized the other ingredients of this complex in the table below:
\begin{center}
\begin{tabular}{lcccc|c|rcccc}\hline
&&&&&&&&&&\\
${\mathcal D}\hspace{-3mm}$&:\!\!\!&$\mathcal{E}(1)$&$\to$&$\odot^2T^*M(1)$&&$\ \ {\mathcal D}^*\hspace{-3mm}$&$:\!\!\!$&$ \odot^2T^*M(-1)~$&$\to$&$\mathcal{E}(-5)$\\[.5mm]
&&\rotatebox{90}{$\in$}&&\rotatebox{90}{$\in$}&&&&\rotatebox{90}{$\in$}&&\rotatebox{90}{$\in$}\\[-1mm]
&&$\sigma$&$\mapsto$&$\big(\nabla_a\nabla_b+Q_{ab}\big)\sigma$&&&&$\gamma_{ab}$&$\mapsto$&$\big(\nabla^a\nabla^b+Q^{ab}\big)\gamma_{ab}$\\[5mm]
$\mathcal{L}_0\hspace{-3mm}$&$:\!\!\!$&$\mathcal{E}(1)$&$\to$&$\mathcal{T}_A$&&$\ \ \mathcal{L}^0\hspace{-3mm}$&$:\!\!\!$&$\mathcal{T}_A(-4)$&$\to$&$\mathcal{E}(-5)$\\[.5mm]
&&\rotatebox{90}{$\in$}&&\rotatebox{90}{$\in$}&&&&\rotatebox{90}{$\in$}&&\rotatebox{90}{$\in$}\\[-1mm]
&&$\sigma$&$\mapsto$&$\left(\sigma,\nabla\sigma\right)$&&&&$\left(\rho, \mu_a\right)$&$\mapsto$&$\frac \Lambda 3 \rho-\nabla^a \mu_a$\\[5mm]
$\mathcal{L}_1\hspace{-3mm}$&$:\!\!\!$&$\odot^2T^*M(1)$&$\to$&$\Wedge^1(\mathcal{T}_A)$&&$\ \ \mathcal{L}^1\hspace{-3mm}$&$:\!\!\!$&$\Wedge^1(\mathcal{T}_A(-2))$&$\to$&$\odot^2T^*M(-1)$\\[.5mm]
&&\rotatebox{90}{$\in$}&&\rotatebox{90}{$\in$}&&&&\rotatebox{90}{$\in$}&&\rotatebox{90}{$\in$}\\[-1mm]
&&$\varphi_{ab}$&$\mapsto$&$\left(0,\varphi_{ab}\right)$&&&&$\left(\xi_a,\nu_{ab}\right)$&$\mapsto$&$\nu_{(ab)}$\\[5mm]
\hline
\end{tabular}
\end{center}
\vspace{5mm}
A straightforward   computation shows that
\begin{equation}\label{projPM}
\boldsymbol{M}^{\nabla }\mathcal{L}_1(\varphi_{ab})=\left(\begin{array}{c}-\nabla.{\!\bm \varphi\!}+{\!\bm \nabla\!\!}\bar\varphi \\[2mm]
\Delta{\!\bm \varphi\!\!}_{ a}-\nabla^b{\!\bm \nabla\!\!}\varphi_{ab}-{\!\bm W\!\!}^{cd}{}_a\varphi_{cd}\end{array}
\right) \, .
\end{equation}
Here, as earlier, $({\!\!\bm \varphi\!\!}_a,{\!\bm W\!\!}^{cd}{}_a)$ denote the soldered one-forms made from $\varphi_{ba}$ and $\W_{b}{}^{cd}{}_a^{\phantom{d}}$, while $\nabla$ is the Levi-Civita coupled exterior derivative and $\Delta:=\nabla^a\nabla_a$ is the Bochner Laplacian. By construction (see Theorem~\ref{Bigtheorem}), the quantities in the  above enjoy the gauge invariance~\eqref{spin2gaugeinvariance}.
The top slot is exactly the constraint~\ref{spin2constraint}. The bottom slot has both a symmetric and antisymmetric piece.
The latter vanishes on Einstein backgrounds using the constraint.
The symmetric piece (which is the image of ${\mathcal L}^1$), again modulo the constraint, is precisely the PM equation of motion~\ref{PM2eom}.
Finally, note that $\df^\nabla {\mathcal L}_1(\varphi_{ab})$ produces a tractor for which only the bottom slot, $\boldsymbol{\nabla}\boldsymbol{\varphi}_b$, is nonvanishing, and hence both projectively and gauge invariant. This quantity is the PM curvature found in~\cite{DeserEM}.

\subsection{Higher spin Einstein obstructions}

Spins greater than two do not enjoy the special status of metric fluctuations which arise as the linearization of a consistent interaction theory. Their  propagation in constant curvature spaces has been studied in detail~\cite{Deser:2001pe,Deser:2001us,Deser:2001wx,Deser:2001xr,Deser:2003gw,Deser:2004ji,Deser:2013xb}. In this section we employ the BGG machine to characterize the obstruction to  putative maximal depth PM spin three couplings to  Einstein backgrounds. 

Again our starting point is the projectively invariant operator ${\mathcal D}$ given in~\eqref{spin3proj}.
However, now reconsider the diagram: 
$$
\begin{diagram}
\mathcal{T}_{(AB)}&\rTo^{\, \dfs^{\nabla}\ } &
\Wedge^1(\mathcal{T}_{(AB)})\\
\uTo_{\mathcal{L}_0} & &\uTo_{\mathcal{L}_1 }\\
\mathcal{E}(2)&\rTo^{\mathcal D\ } & \odot^3T^*M(2)
\end{diagram}
$$
When the structure ${\p}$ is  projectively flat 
the above diagram is commutative, while 
to construct the corresponding detour complex we must require the structure to be projective Einstein; these two conditions together yield constant curvature theories.

Relaxing the projectively flat condition, we may still construct a BGG detour, at the cost of replacing the space  $\odot^3T^*M(2)$
by the reducible bundle $(T^*M\otimes \odot^2T^*M)(2)$. We chronicle this new detour complex in Proposition~\ref{reducibledetour} at the end of this section. There is however, an alternate---well-known in the theory of prolongations for overdetermined systems~\cite{eastwood,EM})---method to maintain commutativity of the above diagram. This is achieved by replacing~$\nabla$ by the {\it prolongation connection}~$\widetilde \nabla$, defined such that the diagram above  commutes. Of course, 
this new connection may not, in general,  solve the Yang--Mills equations required for the detour sequence to be a complex. This gives a tight characterization of the obstruction to higher spin, Einstein,  PM models.

We are now tasked with finding the prolongation connection $\widetilde \nabla$ on the vector bundle~${\mathcal T}_{(AB)}$ obeying
$$
\df^{\widetilde \nabla }\circ {\mathcal L}_0=
{\mathcal L}_1\circ {\mathcal D}\, .
$$
For that computation we focus, for simplicity 
on Einstein projective structures. In the Einstein scale we find
\begin{equation}\label{DoL}
 \Big(\df^{ \nabla }\circ {\mathcal L}_0-
{\mathcal L}_1\circ {\mathcal D}\Big)(\sigma)=\frac13\, Z^a Z^b
 {\!\bm W\!\!\!}_{(ab)}{}^d \nabla_d \sigma
 =:q\Big(\frac13
 W_{c(ab)}{}^d \nabla_d \sigma\Big)\, .
\end{equation}
Here we have used the map  $Z:T^*M(1) {\longrightarrow} \cT^*$ in the composition series~\eqref{composition} to define the insertion operator~$q$ into the bottom slot of  $\Wedge^1\big(\cT_{(AB)}\big)$:
\begin{equation}\label{projector}
\begin{array}{ccccc}
q&\!\!:\!\!&
\big(T^*M\otimes \odot^2T^*M\big)(2)&\longrightarrow& \Wedge^1\big(\cT_{(AB)}\big)\\[2mm]
&&
\rotatebox{90}{$\in$}&&
\hspace{0cm}\rotatebox{90}{$\in$}\\[1mm]
&&\alpha_{cab}&\longmapsto&
\begin{pmatrix}
0\\[1mm]
0\\[0mm]
{\!\bm \alpha\!\!}_{ab}
\end{pmatrix}
\end{array}
\end{equation} 
The canonical projection $\ker(\iota(X))\cap\Wedge^1\big(\cT_{(AB)}\big)\to \big(T^*M\otimes \odot^2T^*M\big)(2)$ is denoted by~$q^\star$.
Note that for Einstein projective structures, $q^\star$ is the formal adjoint of $q$.

On sections of~${\mathcal T}_{(AB)}$
the prolongation connection  is given by
$$
\widetilde\nabla_a \begin{pmatrix}
\tau\\ v_b \\ w_{bc}
\end{pmatrix}
:=\begin{pmatrix}
\nabla_a \tau -2 v_a\\[1mm] 
\nabla_a v_b +Q_{ab} \tau - w_{ab}\\[1mm] 
\nabla_a w_{bc} + 2Q_{a(b} v_{c)}
-\frac23 W_{a(bc)}{}^d v_d
\end{pmatrix}\, .
$$
Hence we have by now established the following result:
\begin{lemma}
The following diagram commutes for projective Einstein structures:
$$
\begin{diagram}
\mathcal{T}_{(AB)}&\rTo^{\, \dfs^{\widetilde \nabla}\ } &
\Wedge^1(\mathcal{T}_{(AB)})\\
\uTo_{\mathcal{L}_0} &
\ \ \ \ \scalebox{1.9}{$\circlearrowleft$}
 &\uTo_{\mathcal{L}_1 }\\
\mathcal{E}(2)&\rTo^{{\mathcal D}\ } & \odot^3T^*M(2)
\end{diagram}
$$
\end{lemma}

\begin{remark}
The above result can be trivially extended to general projective structures~$\p$.
\end{remark}

The curvature~$\widetilde\Omega\in \Wedge^2\big({\rm End}({\mathcal T}_{(AB)})\big)$ of the prolongation connection~$\widetilde \nabla$ acting on a section of ${\mathcal T}_{(AB)}$  is given by:
$$
\widetilde \Omega\circ \begin{pmatrix}
\tau\\[1mm] v_a \\[1mm] w_{ab}
\end{pmatrix} =
\begin{pmatrix}
0\\[1mm] 0\\[1mm]
-\frac23 ({\!\bm \nabla\!\!} {\!\bm W\!\!\!}_{(ab)}{}^c)v_c+2{\!\bm W\!\!\!}_{(a}{}^c w_{b)c}+\frac23 {\!\bm W\!\!\!}_{(ab)}{}^c {\!\bm w\!\!}_c
\end{pmatrix}\, .
$$
As usual, ${\!\bm w\!\!}_b$ here denotes the one-form obtained by soldering $w_{ab}$.
This curvature does not obviously  obey the Yang--Mills equation~\eqref{YM} (replacing, of course $\nabla\to \widetilde\nabla$) unless  the structure is projectively flat. 

\begin{remark}
The detour operator  arises as the second variation of a Yang--Mills action principle~\cite{GSS}. Hence, one might consider 
other gauge invariant action principles whose first and second variations would modify, respectively, the Yang--Mills equation and  detour operator. Whether there exists such an action principle for which the connection~$\widetilde \nabla$ is ``Yang--Mills'', is an open question.
\end{remark}

\begin{remark}
It is interesting to note that some results do exist for massive spin three theories coupled to curved backgrounds. In particular, Zinoviev has given an action principle purported to describe massive spin three excitations on general Ricci flat spacetimes~\cite{Zinoviev:2008ck}. This action  depends on a quartet of totally symmetric, rank $(0,1,2,3)$, tensor fields and enjoys accompanying gauge invariances. For Minkowski backgrounds, the rank $(0,1,2)$ ``St\"uckelberg''
fields of that model can be algebraically gauged away leaving  a minimal field content that can be compared with that found within our BGG framework.
This might be taken as evidence for the existence of curved versions of the spin three PM BGG detour complex. However we note that for general Ricci flat backgrounds, the auxiliaries in the approach of~\cite{Zinoviev:2008ck} can no longer be algebraically gauged away.
\end{remark}

Finally, as promised, we describe a new detour complex
whose long operator is defined for reducible bundles. First we recompute~\eqref{DoL} for general projective structures:
$$
\Big[\df^{ \nabla }\circ {\mathcal L}_0-
{\mathcal L}_1\circ {\mathcal D}\Big](\sigma)=
\frac13 \big({\!\!\bm W\!\!\!}_{(ab)}{}^c \nabla_c - 2{\! \bm C\!\!}_{(ab)}\big)\sigma\, .
$$
This gives a  projectively invariant operator
$$
\begin{array}{ccccc}
{\mathcal D}_{_{\scalebox{.25}{\yng(2,1)}}}&\!:\!&
\ce(2)&\longrightarrow &
\!\!\!\!\!\!T_{_{\scalebox{.25}{\yng(2,1)}}}^*M(2)\\[2mm]
&&
\rotatebox{90}{$\in$}&&
\hspace{-.4cm}\rotatebox{90}{$\in$}\\[.5mm]
&&\sigma&\longmapsto&
\big({W}_{a(bc)}{}^d \nabla_d - 2{ C}_{a(bc)}\big)\sigma\, .
\end{array}
$$
The above is trivial for projectively flat structures.
It is the mixed symmetry part of the map
$$
\begin{array}{ccccc}
\widetilde {\mathcal D}&\!:\!&\ce(2)&\longrightarrow& 
\big(T^*M\otimes \odot^2T^*M\big)(2)\\[2mm]
&&
\rotatebox{90}{$\in$}&&
\hspace{-.4cm}\rotatebox{90}{$\in$}\\[.5mm]
&&\sigma&\longmapsto&
\frac 1 2 \nabla_c \nabla_a \nabla_b\sigma+{Q}_{ c(a}\nabla_{b)}\sigma+ \nabla_c (Q_{ab} \sigma)
\end{array}
$$
We use this map to establish the following result.
\begin{lemma}
For any projective structure~$\p$, the following diagram commutes:
$$
\begin{diagram}
\mathcal{T}_{(AB)}&\rTo^{\, \dfs^{ \nabla}\ } &
\Wedge^1(\mathcal{T}_{(AB)})\\
\uTo_{\mathcal{L}_0} &
\ \ \ \ \scalebox{1.9}{\ \ $\circlearrowleft$}
 &\uTo_{q }\\
\mathcal{E}(2)&\rTo^{\widetilde{\mathcal D}\ } & \big(T^*M\otimes \odot^2T^*M\big)(2)
\end{diagram}
$$
\end{lemma}
We have by now gathered together the main parts of the BGG detour machine and assemble them in the next proposition:
\begin{proposition}\label{reducibledetour}
Let $\p$ be an Einstein projective structure and denote the formal adjoint of $\widetilde {\mathcal D}$ by $\widetilde{\mathcal D}^\star$.  Then, calling $$M:=q^\star\!\circ\!\boldsymbol{M}^\nabla\! \circ q\, ,$$ 
we have the following complex
\begin{equation*}
 0\longrightarrow C^\infty (M) \xrightarrow{\,\,\,\widetilde {\mathcal D} \,\,\,}\big(T^*M\otimes \odot^2T^*M\big)
\xrightarrow{\,\,\,\,\,M \,\,\,\,\,}T^*M\otimes \odot^2T^*M
 \xrightarrow{\,\,\, \widetilde{\mathcal
    D}^\star \,\,\,} C^\infty( M)\longrightarrow 0.
\end{equation*}
\end{proposition}

\begin{proof}
We have constructed $\widetilde{\mathcal D}$ so that the analog of the first square in the diagram~\eqref{cdi} commutes. Moreover the connection $\nabla^\cT$ is Yang--Mills. Thus we are in the situation of 
Theorem~\ref{4.2+}.
\end{proof}

\begin{remark}
The above proposition defines a novel gauge theory on Einstein backgrounds. However, since the image of the gauge operator $\widetilde{\mathcal D}$ is no longer reducible, the irreducible gauge field content is a totally symmetric  {\it and} a mixed symmetry tensor, both of rank three. These are, by construction, subject to a single, scalar, gauge invariance as well as a system of gauge invariant constraints and equations of motion determined by $M^\nabla$. The physical consequences of the mixed symmetry field content is currently unclear.
\end{remark}

\section{Action principles}\label{Action}

Action principles for PM  models were first constructed, for lower spin examples, in~\cite{Deser:2001us}. These were extended to arbitrary spins in~\cite{Zinoviev} by integrating in additional auxiliary fields and then requiring there exist extra gauge invariances removing these leaving the PM systems. Shortly afterwards 
it was realized that action principles could be obtained geometrically by log-radially reducing~\cite{Biswas}  massless systems in one higher dimension~\cite{Hallowell}. As mentioned in the introduction and explicated in Appendix~\ref{LRR}, log-radial reduction  is intimately related to projective geometry. The actions obtained from this reduction were ``metric--like'', meaning  that the field content was arranged in sections of totally symmetric projective tractor bundles. The BGG machine gives what is often called a ``frame-like formulation'' because one deals with tractor-valued differential forms. Frame-like PM action principles have been given in~\cite{Skvortsov:2006at,Zinoviev:2014zka}. These methods introduce progressively more field content 
in order to make actions simple, and possibly amenable to interacting theories~\cite{Vasiliev}.
Conversely, first order, Hamiltonian action principles written in terms of only the physical DoF were introduced in~\cite{Deser:2004ji}.
Here, we give action principles germane to the detour set-up because we are interested in the connection of these systems to projective geometry.

Our key principle is to construct gauge invariant functionals; these are  far from unique. 
Ensuring that gauge invariant equations also imply a gauge invariant system of constraints as their integrability conditions  singles out
the correct action principle.

The operator $\boldsymbol{M}^\nabla$ annihilates the gauge operator $\df^{\nabla}$. Therefore  the functional, defined on projective Einstein structures,  
\begin{equation}\label{SV}
S[\!\Vf]:=\langle \Vf,\boldsymbol{M}^\nabla \!\Vf\,\rangle\, ,
\end{equation}
where $\Vf\in \Wedge^1(\odot^{k} {\mathcal T}^*)$,
is gauge invariant for any spin $s=k+1$.
Here the pairing $\langle {\!\bm U\!\!},\!\Vf\rangle $ denotes
$$
\langle {\!\bm U\!\!},\!\Vf\rangle:=
\int_M g^{ab} \, U_{aA_1\ldots A_k} 
V_{bB_1\ldots B_k} \, 
H^{A_1B_1}\cdots H^{A_kB_k} 
\, .$$ 

We now  focus (partly for simplicity, but also because this case is special) on spin two. 
In this case, calling 
$$
V_{aA}:=\begin{pmatrix}
v_a\\ \psi_{ab}
\end{pmatrix}\ , 
$$
the gauge invariance of~\eqref{SV}
reads
\begin{equation}\label{gauge2}
\left\{
\begin{array}{rcl}
\delta v_a&=&\nabla_a\varepsilon -\xi_a\\[2mm]
\delta \psi_{ab}&=&\nabla_a\xi_b+Q_{ab}\varepsilon\, ,
\end{array}\right.
\end{equation}
where $\begin{pmatrix}\xi_a & \varepsilon\end{pmatrix}\in \cT_A$.
Thus, since the one-form $v_a$ enjoys an algebraic gauge invariance, it can be gauged away. However, 
there is still  the freedom to further gauge  transform $\psi$ along the locus $\xi_a=\nabla_a \varepsilon$, yielding
\begin{equation*}
\delta \psi_{ab}=\big(\nabla_a \nabla_b + Q_{ab}\big) \varepsilon\, .
\end{equation*}
Here only the symmetric part of $\psi_{ab}$ transforms (while its antisymmetric piece $\psi_{[ab]}$ is gauge inert). Therefore,  in the action~\eqref{SV}, we may set
\begin{equation*}
\psi_{ab}=\varphi_{ab}\, ,\quad v_a=0\, ,\end{equation*}
where $\varphi_{ab}$ is symmetric, and so obtain a functional $S[{\mathcal L}_1(\varphi)]$ invariant under the gauge symmetry $\delta \varphi = {\mathcal D} \varepsilon$, which is, of course, the PM gauge symmetry~\ref{1BGGspin2}.

We must still account for the gauge invariant constraints. These are encoded by the Bianchi identity corresponding to the  gauge symmetry~\eqref{gauge2} of the original action~\eqref{SV}, $$\deltaf^\nabla \boldsymbol{M}^\nabla \Vf=0\, .$$
Denoting the two equations of motion of the action~\eqref{SV}
$$
\boldsymbol{M}^\nabla V_a:=\begin{pmatrix}
{\mathcal G}_a^v\\[1mm]
{\mathcal G}^\psi_{ab}
\end{pmatrix}\, ,
$$
the above Bianchi identity implies
\begin{equation}\label{not Ricci}
\nabla^a{\mathcal G}^\psi_{ab}+{\mathcal G}_b^v=0\, .
\end{equation}
Now observe that the equation of motion of the PM gauge invariant action $S[{\mathcal L}_1(\varphi)]$
is related to that of $S[V]$ by
$$
{\mathcal G}^\varphi_{ab}={\mathcal G}^\psi_{(ab)}\raisebox{-.7mm}{\big|}{}_{{}_{\scriptstyle \psi=\varphi,v=0}}\, .
$$
Applying the  Bianchi identity~\eqref{not Ricci} to the field configuration $V={\mathcal L}_1 (\varphi)$, we thus learn
\begin{equation*}
\nabla^a{\mathcal G}^\varphi_{ab}+\Big(\nabla^a{\mathcal G}^\psi_{[ab]}+{\mathcal G}_b^v\Big)\raisebox{-.7mm}{\big|}{}_{{}_{\scriptstyle \psi=\varphi,v=0}}=0\, .
\end{equation*}
The second and third terms above are an integrability condition of the PM equation of motion and thus ought yield the divergence constraint. Indeed, the last of these is $$
{\mathcal G}_a^v\raisebox{-.7mm}{\big|}{}_{{}_{\scriptstyle \psi=\varphi,v=0}}=
X^A \big[\boldsymbol{M}^\nabla {\mathcal L}_1(\varphi)\big]_{aA}=-\nabla.\varphi_a + \nabla_a \bar\varphi\, ,
$$
which is precisely the divergence constraint.
However, we still need to remove the term~$\nabla^a{\mathcal G}^\psi_{[ab]}$. It is 
proportional to the divergence of the variation of the square of the  divergence constraint
so can be removed by adding such a  term to the action. This yields the PM  spin two action principle 
$$
S[\varphi]=\langle {\mathcal L}_1 (\varphi), \boldsymbol{M}^\nabla {\mathcal L}_1( \varphi)\rangle -\frac12
\langle \iota(X) \boldsymbol{M}^\nabla {\mathcal L}_1 (\varphi),\iota(X) \boldsymbol{M}^\nabla {\mathcal L}_1 (\varphi)\rangle\, .
$$
Note that $F_{abc}:=\nabla_a \varphi_{bc}-\nabla_b \varphi_{ac}$ is a gauge invariant curvature for the PM field $\varphi$ (the PM action in terms of this curvature first appeared in~\cite{DeserEM}). The top slot of the image of~$\boldsymbol{M}^\nabla {\mathcal L}_1$, denoted above by $\iota(X) \boldsymbol{M}^\nabla {\mathcal L}_1 (\varphi)$, is the both the trace of this curvature and the divergence constraint. Finally, we note, that by construction, the variation of the above action yields the PM equation of motion~\eqref{PM2eom}.

\section*{Acknowledgements}
The authors gratefully acknowledge support from the Royal Society of
New Zealand via Marsden Grant 13-UOA-018 and the UCMEXUS-CONACYT grant
CN-12-564.  E.L. and A.W. thank the University of Auckland for
hospitality.  E.L. acknowledges partial support from SNF Grant
No. 200020-149150/1.  A.W. was also supported by a Simons
Foundation Collaboration Grant for Mathematicians ID 317562.

\appendix

\section{Log-radial reduction}\label{LRR}

Actions and gauge operators for PM models with arbitrary spin were
originally computed using the log-radial reduction
method~\cite{Biswas} in~\cite{Hallowell}. This technology is closely
related to projective geometry. Here we sketch this link and extend
the log-radial reduction technique to Ricci flat metric cone spaces,
and therefore Einstein backgrounds.  Note that, in a general
projective setting, this ambient space is known as the Thomas cone
space, in view of~\cite{T}, and the link between this, (the metric
cone~\cite{Gallot}) and the projective tractor connection and calculus  treated
in~\cite{ageom,GPW}.

We wish to consider a Ricci flat, $(n+1)$-dimensional ``ambient'' manifold  $(\tilde M,ds^2)$. For that we make the  metric
ansatz
$$
ds^2=e^{2u}\big(du^2+d\Omega^2)\, .
$$
Here we focus on Riemannian signature (but other signatures can be treated similarly). 
If  the $n$-dimensional metric $d\Omega^2$ is $u$-independent, this metric ansatz implies that the vector field $\xi:=\frac{\partial}{\partial u}$ is a homothety, namely
$$
\pounds_\xi ds^2  = 2 ds^2\, . 
$$
Furthermore,
the diffeomorphism $r=e^u$ shows that~$ds^2$ is a cone over the $n$-dimensional, constant $u$, manifold~$M$ with metric $d\Omega^2$. (Note that for the ambient metric $\tilde g=ds^2$, we have $\nabla^{\tilde g}_A \xi^B=\delta_A^B$; see~\cite{Gibbons} for the
proof that this condition implies $\tilde g $ is a cone metric.)
Thus $u$ is a logarithm of the ``radius'' $r$, hence the terminology ``log-radial reduction''. Further,  requiring that the metric $d\Omega^2$ is Einstein with scalar curvature~${\sf Sc}=n(n-1)$ makes $ds^2$ Ricci flat. For the metric $d\Omega^2$, we have $\Lambda=\frac{(n-1)(n-2)}{2}$. Choosing a frame
$$
E^A=e^u\begin{pmatrix}
{\!\bm e\!\!}^a\\
du
\end{pmatrix}\, ,
$$
where $d\Omega^2 = {\!\bm e\!\!}^a \odot {\! \!\bm e\!\!}_a$, the Levi-Civita connection of $ds^2$, pulled back to $M$ acts as
$$
{\!\bm \nabla\!\!}^*\begin{pmatrix}
v^a \\[1mm] \rho
\end{pmatrix}
=
\begin{pmatrix}
{\!\bm \nabla\!\!} \rho\ -{\!\bm v\!\!}\\[1mm]
{\!\bm \nabla\!\!} v^a + {\!\bm e\!\!}^a \rho
\end{pmatrix}={\!\bm \nabla\!\!}^{\cT} \begin{pmatrix}
v^a \\ \rho
\end{pmatrix} \, .
$$
Here we employ $\boldsymbol{e}^a$ as the soldering form and recognize the above as exactly the tractor connection~\eqref{gradients}  for  Einstein projective structures  whose Schouten tensor obeys~$Q_{ab}=g_{ab}$.
Specializing to homogeneous sections  $V^A\in T\tilde M$, with homogeneity condition
$$
\frac{\partial}{\partial u} V^A = w V^A\, ,
$$ 
we further see that $\nabla_A V^B$ evaluated along a constant $u$ hypersurface~$M$, yields precisely  the Thomas $D$-operator acting on $V^B$, so  $D_AV^B$, as given in~\eqref{ThomasD}. 

Massless higher spins in flat space were first written down in~\cite{Fronsdal,Curtright}.
They can be described by a totally symmetric tensor $\Phi_{A_1\ldots A_s}$ subject to the gauge symmetry
\begin{equation}\label{higher_spin_gauge}
\Phi_{A_1\ldots A_s}\sim 
\Phi_{A_1\ldots A_s}+
\nabla_{(A_1}\Xi_{A_2\cdots A_s)}\, ,
\end{equation}
where the fields and gauge parameters obey trace conditions 
$$\Phi_A{}^A{}_B{}^B{}_{A_5\ldots A_s}=0=\Xi_A{}^A{}_{A_3\ldots A_{s-1}}\, .$$
Specializing to distinguished homogeneities, the gauge transformation~\eqref{higher_spin_gauge} is known to encode the PM gauge operators~\cite{Hallowell}. Let us explicate this for the case of spin $s=2$. We now identify $\Phi_{AB}$ with a section of $\ct_{(AB)}(w)$ of homogeneity $w$ and $\nabla_A$ with the Thomas $D$-operator. 
The gauge parameter is thus a section of $\ct_{A}(w+1)$ and the gauge 
transformation becomes (denoting $\Xi_A:=\begin{pmatrix}\xi_a & \varepsilon\end{pmatrix}$)
$$
D_{(A} \Xi_{B)}
=
\begin{pmatrix}
\nabla_a \xi_b + g_{ab} \varepsilon & \nabla_a \varepsilon + w \xi_a\\[2mm]
\nabla_b \varepsilon + w \xi_b& (w+1)\varepsilon
\end{pmatrix}\, .
$$
The PM model appears at homogeneity $w=-1$. Denoting
$$
\cT_{(AB)}(-1)\ni \Phi_{AB}=:\begin{pmatrix}
\tau\\[1mm]
v_a\\[1mm]
\varphi_{ab}
\end{pmatrix}\, ,
$$
we have the PM gauge transformations
\begin{equation}
\left\{
\begin{array}{l}
\delta \tau=0\\[1mm]
\delta v_a=\nabla_a \varepsilon - \xi_a\\[1mm]
\delta \varphi_{ab}=\nabla_a \xi_b + g_{ab}\varepsilon\, .
\end{array}
\right.
\end{equation}
Because the field $\tau$ is gauge inert, it can be consistently 
set to zero.
Moreover, in our Einstein projective setting $Q_{ab}=g_{ab}$, so we see that the above transformation exactly matches those produced by the BGG detour machine  in~\eqref{gauge2}.

\end{document}